\documentclass[12pt,jounral,onecolumn]{IEEEtran}

\usepackage{graphics,
           psfrag,
           epsfig,
           amsthm,
           cite,
           amssymb,
           url,
           dsfont,
           subfigure,
           algorithm,
           algorithmic,
           balance,
           enumerate,
           color,
           setspace,
           multirow
}
\usepackage{amsmath}

\newtheorem{lemma}{Lemma}

\newtheorem{theorem}{Theorem}

\newcommand{\sref}[1]{Section~\ref{#1}}
\newcommand{\appref}[1]{Appendix~\ref{#1}}
\newcommand{\fref}[1]{Figure~\ref{#1}}

\newcommand{\cref}[1]{Constraint~\ref{#1}}
\newcommand{\thref}[1]{Theorem~\ref{#1}}
\newcommand{\lref}[1]{Lemma~\ref{#1}}
\newcommand{\tref}[1]{Table~\ref{#1}}

\newcommand{\ignore}[1]{}

\IEEEoverridecommandlockouts
\begin{document}
\doublespacing

\title{Coding Opportunity Densification Strategies \\ for Instantly Decodable Network Coding}
\author{Sameh~Sorour,~\IEEEmembership{Member,~IEEE,}
        Shahrokh~Valaee,~\IEEEmembership{Senior Member,~IEEE}
\thanks{The authors are with the Edward S. Rogers Sr. Department of Electrical and Computer Engineering,
    University of Toronto, 10 King's College Road, Toronto, ON, M5S 3G4, Canada,
    e-mail:\{samehsorour, valaee\}@comm.utoronto.ca.}\thanks{This work has been submitted to the IEEE for possible publication. Copyright may be transferred without notice, after which this version may no longer be accessible.}}

\maketitle

\IEEEoverridecommandlockouts

\begin{abstract}
In this paper, we aim to identify the strategies that can maximize and monotonically increase the density of the coding opportunities in instantly decodable network coding (IDNC).Using the well-known graph representation of IDNC, first derive an expression for the exact evolution of the edge set size after the transmission of any arbitrary coded packet. From the derived expressions, we show that sending commonly wanted packets for all the receivers can maximize the number of coding opportunities. Since guaranteeing such property in IDNC is usually impossible, this strategy does not guarantee the achievement of our target. Consequently, we further investigate the problem by deriving the expectation of the edge set size evolution after ignoring the identities of the packets requested by the different receivers and considering only their numbers. We then employ this expected expression to show that serving the maximum number of receivers having the largest numbers of missing packets and erasure probabilities tends to both maximize and monotonically increase the expected density of coding opportunities. Simulation results justify our theoretical findings. Finally, we validate the importance of our work through two case studies showing that our identified strategy outperforms the step-by-step service maximization solution in optimizing both the IDNC completion delay and receiver goodput.
\end{abstract}

\section{Introduction} \label{sec:intro}
\emph{Network coding (NC)} \cite{850663} has shown great abilities to substantially improve transmission
efficiency, throughput and delay over broadcast erasure channels. The design of network coding algorithms, optimizing throughput and delay performances
over single-hop broadcast erasure channels, has recently been an intensive area of research
\cite{4895447,4476183,Drinea2009,4313060,Blasiak2011,Sadeghi2010,4397057,\ignore{5152148,}5072357,4397041,Sundararajan2009}. Some of these works have focused on packet selection for coding in each transmission in order to optimize a certain metric, such as in-order delay \cite{Sundararajan2009} or video quality \cite{5072357,4397041}. Other works have focused more on receiver selection in each transmission to optimize another set of parameters, such as completion and decoding delay \cite{Sadeghi2009,Sadeghi2010,ICC10,GC10}. These works have considered a subclass of network coding known as \textit{instantly decodable network coding (IDNC)}, in which coded packets must be decoded at their reception instant and cannot be stored for future decoding. This opportunistic network coding scheme has attracted attention due to its desirable properties, such as fast packet recovery, simple XOR coding and decoding, and no buffer requirements.

In most of these opportunistic network coding and IDNC works, the selection of a coding combination to optimize a desired parameter for a particular transmission does not consider the effect of this selection on resulting \emph{``coding opportunities''} in subsequent transmissions. By a coding opportunity we mean the opportunity of serving two packet requests of two receivers simultaneously by one transmission using network coding. For instance, the proposed online algorithms in \cite{4476183,Sadeghi2009,Drinea2009,Sadeghi2010} and \cite{GC10} have focused on increasing the number of served requests and decoding receivers in each transmission without studying the effect of such approach on the number of remaining coding opportunities for subsequent transmissions. If this selection leads to a very limited number of opportunities, then the sender will no longer be able to send packets that are decodable by many receivers, thus sacrificing the main strength of network coding and its overall performance. Instead, if these algorithms consider the maintenance of a larger number of coding opportunities for subsequent transmissions, they may end-up with an overall better performance.

It was not until very recently, when few works started to give weight to the concept of increasing the coding opportunities in opportunistic network coding \cite{Wang2010,6145511}. In fact, these works have shown that increasing the coding opportunities is of major importance to achieve the capacity in some special cases of erasure channels. Nonetheless, the suggested approaches to increase the coding opportunities in \cite{Wang2010,6145511} are done by buffering non-decodable packets, which makes them not suitable for IDNC. This motivates us to explore the strategy that can play a similar role in IDNC.

To give an example on the importance of maintaining a large number of coding opportunities in IDNC, let us assume a network of 6 receivers $\{r_1,\dots,r_6\}$ that require packets $\{p_1,\dots,p_6\}$, respectively\ignore{ in which receivers $r_1$, $r_2$ and $r_3$ are missing packets $p_1$, $p_2$ and $p_3$, respectively, whereas receiver $r_4$ is missing packets $p_4$, $p_5$ and $p_6$}. Also assume that, due to the side information at the different receivers, the available coding opportunities are $\left\{p_1 \oplus p_2 \oplus p_3, p_1 \oplus p_4, p_2 \oplus p_5, p_3 \oplus p_6\right\}$. This scenario can be presented in the form of the graph depicted in \fref{fig:motivation}.a. Each vertex with indices $(i,j)$ in the graph represent the request of packet $p_j$ by receiver $r_i$. An edge between any two vertices $(i,j)$ and $(k,l)$ represents a coding opportunity given the available side information, such that the coded packet $p_j\oplus p_l$ is decodable at both receivers $r_i$ and $r_k$. Given this representation, the first coding opportunity (i.e. $1\oplus 2\oplus 3$) is represented by the left triangle between the three leftmost vertices, whereas the three other opportunities are represented by the three horizontal edges.
\begin{figure}
  \centering
  \includegraphics[width=0.6\linewidth]{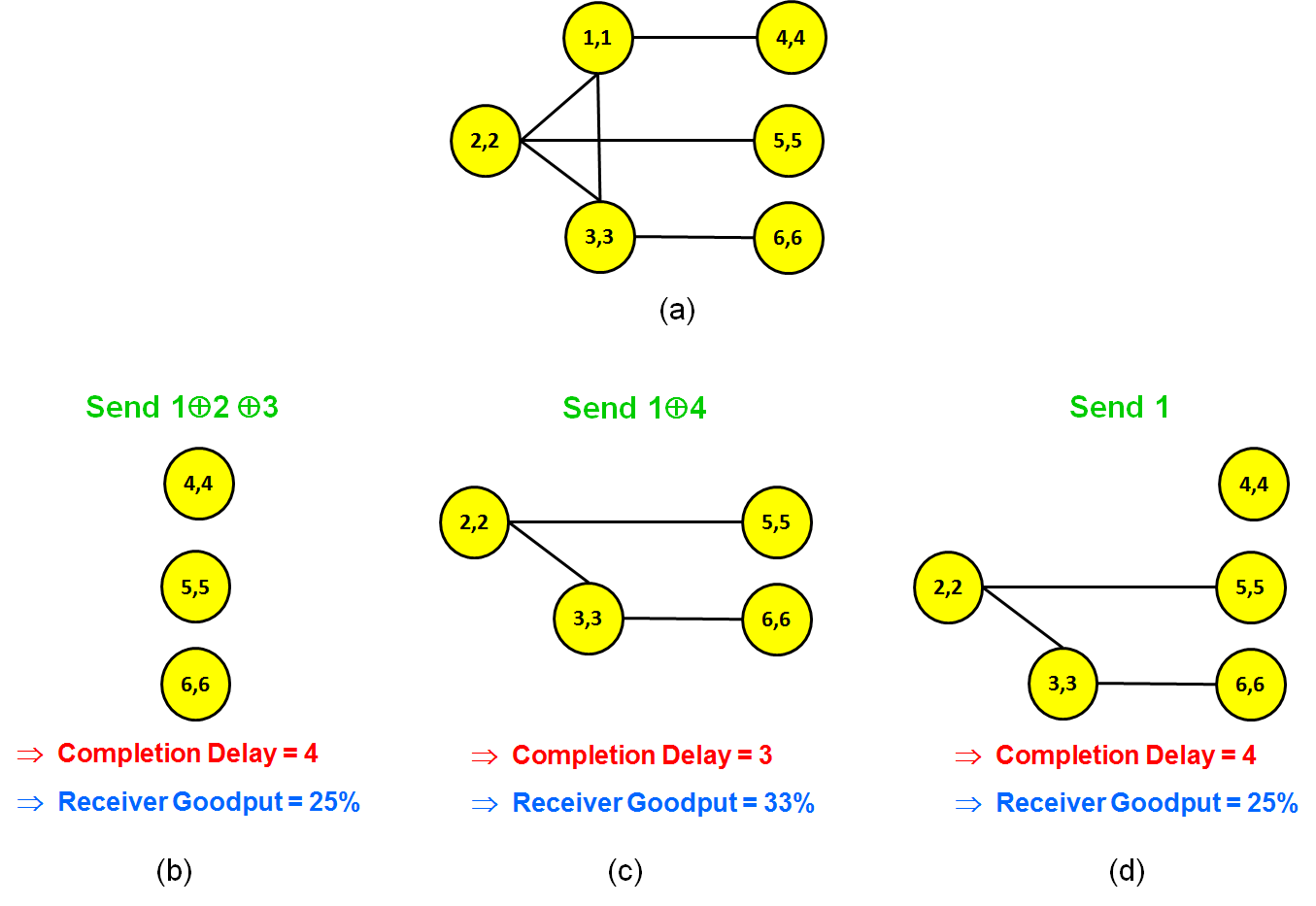}\\
  \caption{Motivating Example}\label{fig:motivation}
\end{figure}
If we follow the philosophy of serving the maximum number of requests in each transmission (similar to the algorithms in \cite{4476183,Sadeghi2009,Drinea2009,Sadeghi2010,GC10}), the sender should first send the packet combination $p_1 \oplus p_2 \oplus p_3$. However, this selection will result in three packets $\{p_4, p_5, p_6\}$ with no coding opportunities between them, as depicted in \fref{fig:motivation}.b. Thus, the sender will require 3 additional transmissions to serve these requests. Defining the \emph{receiver goodput} as the percentage of the sender transmission from which each receiver would benefit, the obtained receive goodput according to the above strategy would be $25\%$. On the other hand, a first transmission serving a smaller number of packet requests (such as $p_1 \oplus p_4$) keeps 3 coding opportunities in the system, as depicted in \fref{fig:motivation}.c, which can be satisfied by only 2 additional transmissions and the resulting goodput will be $33\%$. Consequently, the selection of coded transmissions that preserve a large number of coding opportunities in the system results in a better completion delay and receiver goodput.

Despite its importance illustrated by the above example, the number of coding opportunities may not be expressive in itself. Indeed, a selected transmission may result in a larger number of coding opportunities because it serves very few packet requests, and thus the number of remaining requests will be also large. Thus, this apparently large number of coding opportunities will not be enough to foster efficient combinations between the large number of remaining requests. Consequently, it is important not only to maximize the absolute number of coding opportunities but to mainly maximize their ratio to the number of remaining packet requests. This notion is depicted in Figures \ref{fig:motivation}.c and \fref{fig:motivation}.d. In both cases, the used transmission leads to the same number of remaining coding opportunities. However, the final completion delay and receiver goodput is better for the the case of \fref{fig:motivation}.c as the available number of coding opportunities is serving less number of requests.

Given this observation, we define the \emph{coding density} as the number of actual coding opportunities normalized by the maximum number of coding opportunities that could exist for the same number of packet requests. Consequently, this coding density parameter evaluates the number of coding opportunities with respect to the total number of these requests.\ignore{ In other words, we aim for a modified densification approach from the one in \cite{Pedarsani2008}, in which the number of remaining requests is reduced much faster than the number of opportunities.} In this paper, we aim to answer the following question: \emph{What are the coding strategies in IDNC that can maximize the coding density after each transmission and can result in its continuous increase over the transmission horizon of a frame of broadcast packets?}

To answer this question, we employ the above graph representation of the requests and coding opportunities, and first derive an expression for the exact evolution of the edge set size after the transmission of any arbitrary coded packet. From the derived expressions, we show that targeting commonly wanted packets for all the receivers can maximize the number of coding opportunities. However, guaranteeing such property in IDNC is usually impossible, especially when such packets are exhausted in the first few transmissions. In this case, the performance of this strategy would greatly depend on the receivers served in each transmission, and the exact evolution expression does not give us any intuition regarding this matter. Consequently, we further investigate the problem by deriving the expectation of the edge set size evolution after ignoring the identities of the packets requested by the different receivers and considering only their numbers. From this expression, we show that the best strategy to increase the expected coding density is to serve the maximum number of receivers having the largest number of missing packets and erasure probabilities. We then test both identified strategies and compare them with other well-known IDNC strategies. Finally, we validate the importance of our study and the chosen metrics by presenting two case studies showing the effect of the identified receiver selection strategies on reducing the completion delay and increasing the receiver goodput in IDNC.

It is important to note that this paper is not proposing algorithms to optimize any specific throughput or delay parameter. It is rather a first and independent study of an influential component in network coding, namely the coding density, which has been totally ignored in most works on designing network coding algorithms, despite its clear importance in optimizing long-term parameters compared to per-transmission benefits. The contributions of the paper can be summarized as follows:
\begin{itemize}
\item It derives expressions for the evolution of coding opportunities and density along the transmission of network coded packets.
\item It provides a rigorous analysis of the parameters affecting the evolution of coding opportunities and density.
    \end{itemize}
These contributions can open paths for future efforts in designing more efficient online network coding algorithms, which optimize different performance metrics while taking our coding density analysis into consideration. The paper finally illustrates the importance of coding density by two case studies on completion delay and receiver goodput.

The rest of the paper is organized as follows. In \sref{sec:model}, the system model and parameters are illustrated. We introduce the IDNC graph and our metric of coding density in \sref{sec:IDNC-graph}. In \sref{sec:graph-evolution}, we derive the expression for the exact edge set size evolution and analyze it in \sref{sec:packet-selection-strategy}. We then derive the expected edge set size evolution in \sref{sec:expected-coding-opportunity-evolution} and identify, in \sref{sec:receiver-selection-strategy}, the coding strategy increasing it. Simulation results are illustrated in \sref{sec:simulations}. \sref{sec:case-study} presents two case studies on the effect of our identified strategies optimizing the IDNC completion delay and receiver goodput. Finally, \sref{sec:conclusion} concludes the paper.

\section{System Model and Parameters} \label{sec:model}
The model consists of a wireless sender that is required to deliver a frame of $N$ source packets (denoted by $\mathcal{N}$) to a set of $M$ receivers (denoted by $\mathcal{M}$).
The sender initially transmits the $N$ packets\ignore{ of the frame} uncoded in an \emph{initial transmission phase}. Each sent packet can be successfully received at receiver $i$ with probability $q_i$, which is assumed to be fixed during the frame transmission period. Receivers feed back to the sender a positive one-bit acknowledgement (ACK) for each received packet. Consequently, an overhead of $O(N)$ bits is required for feedback after each transmission. At the end of the initial transmission phase, two sets of packets are attributed to each receiver $i$, representing the feedback state of the network:
\begin{itemize}
\item The \emph{Has} set (denoted by $\mathcal{H}_i$) is defined as the set of packets correctly received by receiver $i$.
\item The \emph{Wants} set (denoted by $\mathcal{W}_i$) is defined as the set of packets that are not yet received by receiver $i$. In other words, $\mathcal{W}_i = \mathcal{N} \setminus \mathcal{H}_i$.
\end{itemize}
The cardinalities of the Has and Wants sets of receiver $i$ are denoted by $\varrho_i$ and $\psi_i$, respectively. After the initial transmission phase, a recovery \ignore{transmission }phase starts, in which the sender exploits the diversity of received\ignore{ and lost} packets to transmit network coded combinations\ignore{ of source packets}. According to the definition of IDNC, these combinations must be either decoded at their reception instant or discarded. The received ACKs at the sender after each transmission are used to update the different sets. This process is repeated until all receivers obtain all the packets.

\section{The IDNC Graph} \label{sec:IDNC-graph}
As depicted in \fref{fig:motivation}, we can represent the receiver requests and all feasible instantly decodable combination among them using a graph, which we will refer to as the IDNC graph.\ignore{ It was first introduced in the context of a heuristic algorithm design solving the index coding problem\ignore{ \cite{4544612,4313060}} \cite{4544612}.} This graph $\mathcal{G}(\mathcal{V},\mathcal{E})$ is constructed by first generating a vertex $v_{ij}$ in $\mathcal{V}$ for each requested packet $j \in \mathcal{W}_i$, $\forall~i\in\mathcal{M}$. Two request vertices $v_{ij}$ and $v_{kl}$ in $\mathcal{G}$ are set adjacent by a coding opportunity edge if one of the following conditions is true:
\begin{itemize}
\item C1: $j = l$ $\Rightarrow$ The two vertices represent the loss of the same packet $j$ by the two \ignore{different }receivers $i$ and $k$.
\item C2: $j\in \mathcal{H}_k$ and $l \in \mathcal{H}_i$ $\Rightarrow$ The requested packet of each vertex is in the Has set of the receiver that induced the other vertex.
\end{itemize}
Given this graph formulation, it is clear that any group of vertices fully adjacent to one another using coding opportunity edges (thus forming a clique in $\mathcal{G}$) can be served by one coding combination including an XOR of the packets identified by these vertices. According to the design of $\mathcal{G}$, we can easily infer that any clique $\kappa$ in $\mathcal{G}$ can include at most one vertex induced by any given receiver to maintain instant decodability. In the rest of the paper, we say that an IDNC packet \emph{targets} a receiver if the corresponding clique includes a vertex belonging to this receiver.

Based on this modeling of coding opportunities as the edges of a graph, we can define the coding opportunity density (or coding density for short) $\rho_c(\mathcal{G})$ for IDNC as the density of its graph $\mathcal{G}$. In graph theory, the graph density is the ratio of the total number of edges in this graph to the number of edges of a complete graph with the same number of vertices. We can express this graph density (and thus coding density) as:
\begin{equation} \label{eq:coding-efficiency}
\rho_c(\mathcal{G}) = \frac{\left|\mathcal{E}\right|}{\frac{1}{2} \left|\mathcal{V}\right|\left(\left|\mathcal{V}\right| - 1\right)\ignore{ - \sum_{i=1}^M \binom{\psi_i}{2}}}
\end{equation}
\ignore{This metric is a good measure of the density of coding opportunities with respect to the number of vertices in the system, as explained in \sref{sec:intro}.} It is obvious that the maximization of $\rho_c(\mathcal{G})$ guarantees a large number of coding opportunities with respect to the number of remaining packet requests (i.e. vertices), and thus a large number of receivers and packet requests can be served simultaneously in each IDNC packet.

From the above expression, we can see that, in order to maximize the coding density in each step, the selected cliques should be able to both maximize the number of edges and minimize the vertex set size. The number of vertices is clearly minimized by serving the maximum number of receivers in each transmission. However, this selection may decrease the coding density if the numerator is significantly reduced. To study this effect, we first need to derive an expression for the edge set size evolution after any arbitrary transmission.

\section{Exact Coding Opportunity Evolution} \label{sec:graph-evolution}
\ignore{In this section, we aim to derive an expression for the evolution of the number of coding opportunities (i.e. the graph's edge set size) when any arbitrary maximal clique is selected. }In order to derive an expression for the edge set size evolution, we start by deriving the expression of the edge set size for any given feedback state.

\subsection{Edge Set Size} \label{sec:exact-edge-set-size}
Theorem 1 introduces the expression of the edge set size of the IDNC graph.
\begin{theorem}\label{th:edge-set-size}
The edge set size for an arbitrary feedback state can be expressed as:
\begin{equation}\label{eq:edge-set-size}
\ignore{
\left|\mathcal{E}\right| = \frac{1}{2}\sum_{i=1}^M \sum_{\substack{k=1\\k\neq i}}^M \left(\psi_{i}\psi_{k} - \left(\psi_i + \psi_k -1\right)\psi_{ik} + \psi_{ik}^2\right) \qquad \qquad (\mbox{where}~\psi_{ik} = \left|\mathcal{W}_i\cap\mathcal{W}_k\right|)\;.}
\left|\mathcal{E}\right| = \frac{1}{2}\sum_{i=1}^M \sum_{\substack{k=1\\k\neq i}}^M \left(\psi_{ik} + \theta_{ik}\theta_{ki}\right),\qquad \mbox{where}\quad\psi_{ik} = \left|\mathcal{W}_i\cap\mathcal{W}_k\right|,\quad \theta_{ik} = \psi_{i} - \psi_{ik},\quad \theta_{ki} = \psi_{k} - \psi_{ik}\;.
\end{equation}
\end{theorem}
\begin{proof}
The proof can be found in \appref{app:edge-set-size}.
\end{proof}
The intuition behind \thref{th:edge-set-size} is illustrated in \fref{fig:pairwise-edges}.
\begin{figure}[t]
\centering
  \includegraphics[width=0.4\linewidth]{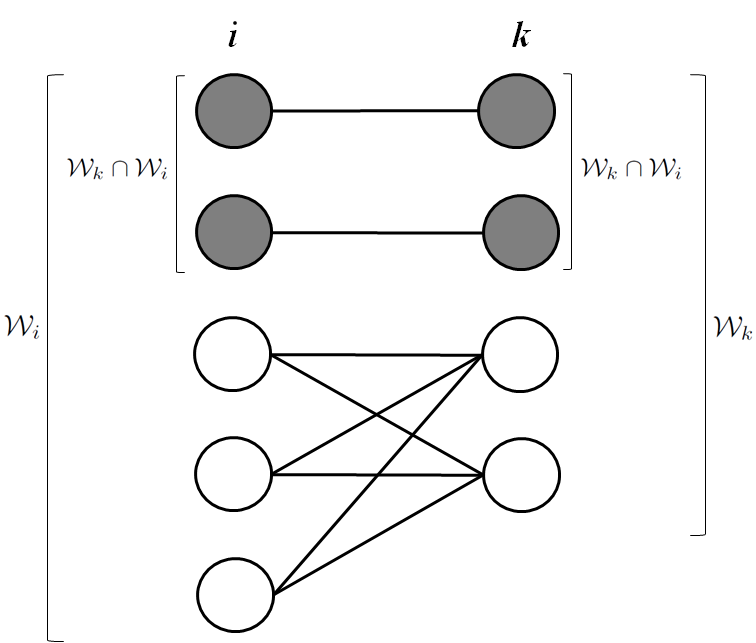}\\
  \caption{Pairwise edges between receivers $i$ and $k$ given their Wants sets. The grey vertices represent the packets that are commonly wanted by both receivers whereas the white vertices represent packets that are wanted by either of them and has been received by the other.}\label{fig:pairwise-edges}
\end{figure}
The figure depicts the number of pairwise edges between receivers $i$ and $k$ (denoted by $Y_{ik}$). As shown, the vertices of $i$ and $k$ can be classified into two sets. The gray vertices represent the vertices of $i$ and $k$ requesting the same packets (i.e. pairs of vertices of $i$ and $k$ with $j=l$). Thus, such vertex pairs are adjacent according to condition C1 in \sref{sec:IDNC-graph}. They cannot be adjacent to other vertices of the opposite receiver as this will violate condition C2. We say that each of these vertices is pairwise restricted by its adjacent vertex at the opposite receiver. Thus, these vertices will contribute to $Y_{ik}$ by $\left|\mathcal{W}_k \cap\mathcal{W}_i\right| = \psi_{ik}$ edges as shown in \fref{fig:pairwise-edges}.

The white vertices represent the mutually unrestricted vertices (i.e. not in $\mathcal{W}_k\cap\mathcal{W}_i$), and can all be connected to each other as a full bipartite subgraph because they all satisfy condition C2. Consequently, they contribute to $Y_{ik}$ by $\left(\psi_i - \psi_{ik}\right)\left(\psi_k - \psi_{ik}\right)$ edges. The final expression in \eqref{eq:edge-set-size} results from summing the $Y_{ik}$'s of all $i$ and $k\neq i$ and dividing by two to remove repetitions.

\subsection{Edge Set Size Evolution} \label{sec:exact-edge-set-evolution}
Before deriving the expression of the edge set size evolution, we will first illustrate the different possibilities of evolution on the pairwise subgraph in \fref{fig:pairwise-edges}, when either or both receivers $i$ and $k$ are targeted with source packets $p_i$ and $p_k$, respectively, in one IDNC transmission. These possibilities are illustrated in \fref{fig:pairwise-edge-evolution}, in which all preserved edges from $Y_{ik}^{(t)}$ to $Y_{ik}^{(t+1)}$ are removed for ease of illustration. The served vertices are marked in black and added (removed) edges are represented by solid (dashed) lines.
\begin{figure}[t]
\centering
  \subfigure[]{\includegraphics[width=0.39\linewidth]{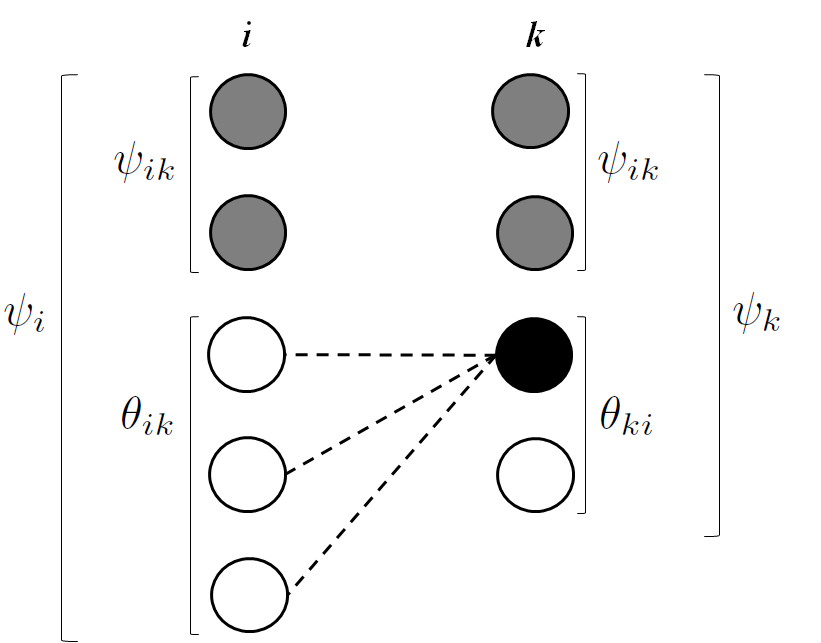}\label{fig:pairwise-edge-evolution-1}}
  \subfigure[]{\includegraphics[width=0.38\linewidth]{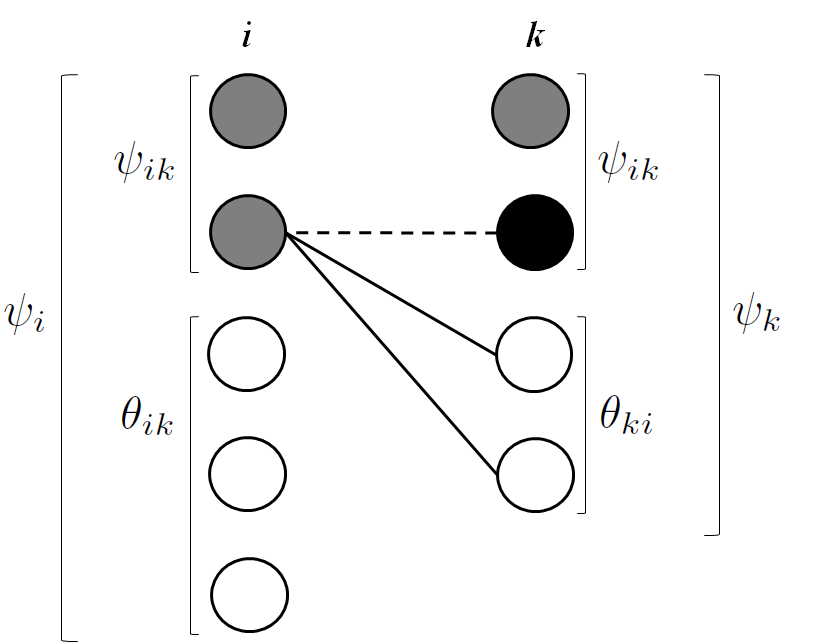}\label{fig:pairwise-edge-evolution-2}}\\
  \subfigure[]{\includegraphics[width=0.39\linewidth]{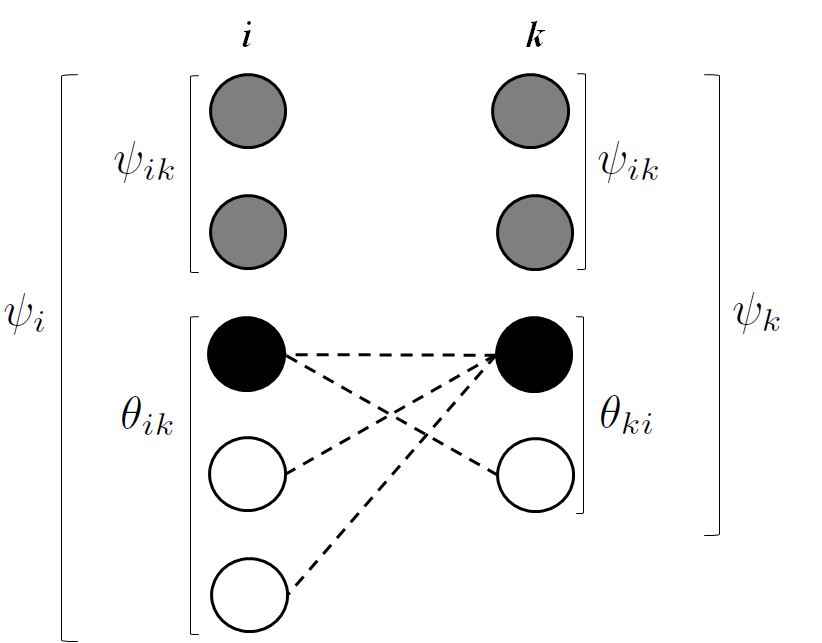}\label{fig:pairwise-edge-evolution-3}}
  \subfigure[]{\includegraphics[width=0.39\linewidth]{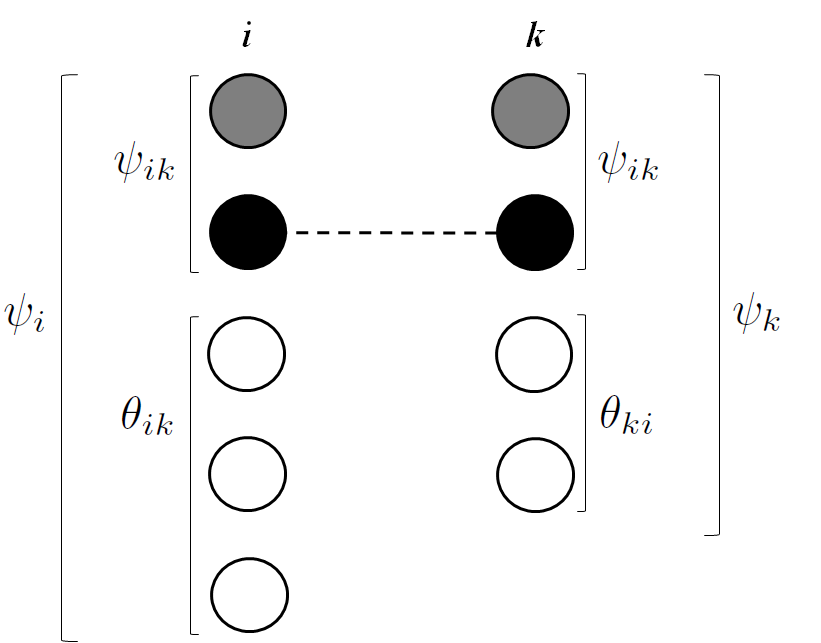}\label{fig:pairwise-edge-evolution-4}}
  \caption{Added (solid) and removed (dashed) edges between the vertices of receivers $i$ and $k$ in case: (a) Only one of the two receivers is targeted with an unrestricted vertex. (b) Only one receiver is targeted with a restricted vertex. (c) Each of the two receivers is targeted with an unrestricted vertex and they both receive. (d) Each of the two receivers is targeted with a restricted vertex and they both receive.}\label{fig:pairwise-edge-evolution}
\end{figure}

In \fref{fig:pairwise-edge-evolution-1}, only one receiver (in this case $k$) is targeted with an unrestricted vertex $v_{kp_k}$ with respect to the other receiver (in this case $i$). If this receiver receives, its vertex $v_{kp_k}$ and all its adjacent edges (which are shown in \fref{fig:pairwise-edges}) will disappear from the graph. Since the opposite vertices to these removed edges are already unrestricted with respect to $k$, they will not gain any additional edges.

In \fref{fig:pairwise-edge-evolution-2}, only one receiver (in this case $k$) is targeted with a restricted vertex $v_{kp_k}$ by the other receiver (in this case $i$). If this receiver $k$ receives, its vertex $v_{kp_k}$ will disappear as well as its edge to the corresponding restricted vertex of $i$ (i.e. $v_{ip_k}$). Consequently, this restricted vertex $v_{ip_k}$ becomes unrestricted with respect to $k$ and thus becomes adjacent to all $k$'s unrestricted vertices with respect to $i$.

The third possibility is when both receivers are targeted with one of their mutually unrestricted vertices. If only one of them receives, we get the case of \fref{fig:pairwise-edge-evolution-1}. If both receivers receive, each of the vertices will behave as in \fref{fig:pairwise-edge-evolution-1} resulting in the evolution in \fref{fig:pairwise-edge-evolution-3}.

The last possibility is when both receivers are targeted with one of their mutually restricted vertices. If only one of them receives, we get the case of \fref{fig:pairwise-edge-evolution-2}. If both receivers receive, they both disappear with their common edge and without addition of new mutual edges, as shown in \fref{fig:pairwise-edge-evolution-4}.

After the description of these possibilities, we can now introduce the following theorem.
\begin{theorem}\label{th:edge-set-size-evolution}
For an arbitrary attempted clique $\kappa$ at time $t$ with a set of targeted receivers $\mathcal{T}$, the edge set size at time $t+1$ after this attempt can be expressed as:
\ignore{
\begin{align}\label{eq:edge-set-size-evolution}
\left|\mathcal{E}^{(t+1)}\right| &= \left|\mathcal{E}^{(t)}\right| \nonumber \\
&+ \frac{1}{2}\sum_{\substack{i\notin\mathcal{T},k\in\mathcal{T}\\p_k\in \mathcal{W}_i}} X_k\hat{\theta}_{ki} - \frac{1}{2} \sum_{\substack{i\notin\mathcal{T},k\in\mathcal{T}\\p_k\notin\mathcal{W}_i}} X_k\theta_{ik}\nonumber\\
& + \frac{1}{2} \sum_{\substack{i\in\mathcal{T},k\notin\mathcal{T}\\p_i\in\mathcal{W}_k}} X_i\hat{\theta}_{ik} - \frac{1}{2} \sum_{\substack{i\in\mathcal{T},k\notin\mathcal{T}\\p_i\notin\mathcal{W}_k}} X_i\theta_{ki}\nonumber\\
& + \frac{1}{2} \sum_{\substack{\{i,k\}\in\mathcal{T}\\p_k\in\mathcal{W}_i}} \Bigg(X_iX_k - X_i\theta_{ki} - X_k\theta_{ik}\Bigg)\nonumber\\
& +  \frac{1}{2} \sum_{\substack{\{i,k\}\in\mathcal{T}\\p_i\in\mathcal{W}_k}} \Bigg(X_i\hat{\theta}_{ik} + X_k\hat{\theta}_{ki}\nonumber\\
& \qquad\qquad\qquad - X_i X_k \left(\hat{\theta}_{ik} + \hat{\theta}_{ki} - X_i X_k\right)\Bigg)\;.
\end{align}
}
\begin{align}\label{eq:edge-set-size-evolution}
&\left|\mathcal{E}^{(t+1)}\right| = \left|\mathcal{E}^{(t)}\right| + \frac{1}{2}\sum_{\substack{i\notin\mathcal{T},k\in\mathcal{T}\\p_k\in \mathcal{W}_i}} X_k\hat{\theta}_{ki} - \frac{1}{2} \sum_{\substack{i\notin\mathcal{T},k\in\mathcal{T}\\p_k\notin\mathcal{W}_i}} X_k\theta_{ik} + \frac{1}{2} \sum_{\substack{i\in\mathcal{T},k\notin\mathcal{T}\\p_i\in\mathcal{W}_k}} X_i\hat{\theta}_{ik} - \frac{1}{2} \sum_{\substack{i\in\mathcal{T},k\notin\mathcal{T}\\p_i\notin\mathcal{W}_k}} X_i\theta_{ki}\nonumber\\
& - \frac{1}{2} \sum_{\substack{\{i,k\}\in\mathcal{T}\\p_i\notin\mathcal{W}_k}} \Bigg(X_i\theta_{ki} + X_k\theta_{ik} - X_iX_k\Bigg)\ignore{\nonumber\\
&} +  \frac{1}{2} \sum_{\substack{\{i,k\}\in\mathcal{T}\\p_i\in\mathcal{W}_k}} \Bigg(X_i\hat{\theta}_{ik} + X_k\hat{\theta}_{ki}- X_i X_k \left(\hat{\theta}_{ik} + \hat{\theta}_{ki} - X_i X_k\right)\Bigg)\;,
\end{align}
where $\hat{\theta}_{ik} = \theta_{ik} - 1$, $\hat{\theta}_{ki} = \theta_{ki} - 1$, and $X_h$ is the reception indicator of receiver $h$, which is equal to 1 if $h$ receives the transmitted packet and zero otherwise.
\end{theorem}
\begin{proof}
The proof can be found in \appref{app:edge-set-size-evolution}.
\end{proof}
We will analyze this obtained exact evolution expression in the next section.

\section{Analysis of the Exact Coding Opportunity Evolution} \label{sec:packet-selection-strategy}
From \eqref{eq:edge-set-size-evolution}, we can prove the following theorem for any two arbitrary receivers.
\begin{theorem}\label{th:two-receivers}
Targeting either one or both receivers $i$ and $k$ with a packet in $\mathcal{W}_{i}\cap \mathcal{W}_{k}$ (i.e. a common wanted packet) at time $t$ results in a greater or equal number of pairwise edges between them at time $t+1$ $\left(\mbox{i.e.}~Y_{ik}^{(t+1)}\right)$ compared to targeting either one or both receivers with packets that are not in $\mathcal{W}_{i}\cap \mathcal{W}_{k}$.
\ignore{is greater than or equal to mmon wanted source packets at time $t$ results in a greater than or equal number of pairwise edges between them at time $t+1$ (i.e. $Y_{ik}^{(t+1)}$) compared to targeting either one or both of them with a non-common wanted source packet them maximizes the number of their pairwise edges at time $t+1$, whether either or both receivers are targeted with this source packet in this transmission}.
\end{theorem}
\begin{proof}
The proof can be found in \appref{app:two-receivers}.
\end{proof}

The above theorem proves that serving a common wanted packet for receivers $i$ and $k$ always results in a larger increase or smaller reduction in the number of their pairwise edges $Y_{ik}^{(t+1)}$, whether both receivers are targeted with this packet or only one of them is. Now, since Equation \eqref{eq:edge-set-size} in \thref{th:edge-set-size} expresses the overall edge set size as a linear addition of these numbers of pairwise edges, then having this property satisfied for all receivers (i.e. the sender transmits a packet that is in the Wants set of all receivers at time $t$) will result in the maximum edge set size that could be achieved at time $t+1$ according to Theorems \ref{th:edge-set-size} and \ref{th:two-receivers}. However, violating this condition for any one receiver $i$ will replace its $Y_{ik}^{(t+1)}$ $\forall~k\in\mathcal{M}\setminus i$ in \eqref{eq:edge-set-size} with smaller values, which will lead to smaller edge set size. Increasing the number of receivers violating this condition will result in further replacements with smaller terms, and thus a larger reduction in the edge set size.

Since the existence of common wanted packets by all the receivers is usually infeasible, one solution could be to serve the packets that are wanted by the maximum number of receivers. We will refer to this strategy as the \emph{Most Wanted Packet Serving (MoWPS) strategy}. We can express this strategy as choosing the maximal clique $\kappa^*$ in each transmission such that:
\begin{equation} \label{eq:MoWPS}
\kappa^* = \arg\min_{\kappa \in \mathcal{G}} \sum_{j|v_{ij}\in\kappa} \left|\mathbf{\Omega}_j\right|^n \qquad\qquad\mbox{s.t.}\quad \mathbf{\Omega}_j = \left\{i\in\mathcal{M}|j\in\mathcal{W}_i\right\}\;,
\end{equation}
where $\mathbf{\Omega}_j$ is the set of receivers wanting packet $j$ and $n$ is a biasing factor. In other words, the MoWPS strategy selects the maximum weight clique in the IDNC, such that the weight of each vertex $v_{ij}$ in the graph is defined by $\mathbf{\Omega}_j$. Thus, this maximum weight clique will include the vertices representing the packets that are wanted by most of the receivers.

However, the MoWPS strategy suffer from two main issues:
\begin{itemize}
\item If a common wanted packet by all the receivers does not exist, which is highly probable, then it will have to serve packets that are unwanted by some of the receivers. Since we cannot infer from \eqref{eq:edge-set-size-evolution} the best selection of such receivers, this strategy may not be able to guarantee the achievement of the best possible performance.
\item Even if there existed packets that are wanted by all or a vast majority of the receivers, this strategy could rapidly deplete such occurrences in the first few transmissions. Consequently, further use of the same strategy will not result in the desired continuous increase in coding density. In other words, MoWPS strategy may succeed in significantly rasing the coding density in the first few transmissions, but also using it further may significantly decrease the coding density in some stages of the recovery transmission phase. We will illustrate and interpret these effects in \sref{sec:simulations}.
\end{itemize}

Since we were not able to identify a clear strategy by analyzing the exact evolution expression of coding opportunities, we need to further investigate the problem, by diminishing the effects that were complicating this exact expression and its analysis. By examining this expression in \eqref{eq:edge-set-size-evolution}, we can see that its complication comes from its great dependence on the identities of the selected packets for transmission and whether they belong to the Has or Wants sets of the different receivers. To eliminate this complication, we will thus derive an expected evolution expression of the edge set size, which only depends on the cardinalities of the receivers' feedback sets and eliminates the dependency on their actual packet contents. We will then re-analyze this new expression to identify a more robust coding strategy to continuously increase the coding density in IDNC. This will be the target of the next two sections.

\ignore{
this finding does not mean that the MoWPS strategy will always increase the coding density at all stages of the recovery transmission phase. The increase of coding density, after any transmission, will be conditioned on the presence of large sets of receivers requesting the same packets for this transmission. When this property is present in the beginning of the recovery phase, transmitting these packets wanted by a large number of receivers will remove a large number of vertices from the graph, increase the number of edges, and thus increase the coding density. At the same time, the number of receivers, still wanting these packets, will considerably decrease after each transmission. Consequently, the MoWPS strategy will result in serving a smaller number of vertices in subsequent transmissions, which may decrease the coding density at some stages of the recovery transmission phase. We will illustrate and interpret these effects in \sref{sec:simulations}.
}

\section{Expected Coding Opportunity Evolution} \label{sec:expected-coding-opportunity-evolution}
In this section, we will derive an expected value representation of the edge set size evolution, which only depends on the cardinalities of the receivers' feedback sets and eliminates the dependency on their actual packet contents. In this case, each Wants set $\mathcal{W}_i$ becomes a random set of packets of size $\psi_i$ drawn from the pool of $N$ original source packets. Consequently, the number of packets that are found in both Wants sets of two receivers $i$ and $k$ (i.e. $\left|\mathcal{W}_i\cap\mathcal{W}_k\right|$) becomes a random variable with hypergeometric distribution. With this approach, the expression of the edge set size evolution will be an expectation given these hypergeometric random variables.\ignore{ To derive this expression, we will first introduce an expression for the expected edge set size for any arbitrary Has and Wants vectors $\boldsymbol{\varrho} = \{\varrho_1,\dots,\varrho_M$\} and $\boldsymbol{\psi} = \{\psi_1,\dots,\psi_M\}$.} We will derive this expression in the following two theorems.
\begin{theorem} \label{th:expected-edge-set-size}
Given the receivers' feedback set cardinalities, the expected edge set cardinality of the graph is equal to:
\begin{equation} \label{eq:expected_edge_set_cardinality}
\mathds{E}\left[\left|\mathcal{E}\right|\right] = \frac{1}{2}\sum_{i=1}^M \psi_i\mathds{E}\left[\Delta_i\right]  = \frac{1}{2}\sum_{i=1}^M \psi_i\left\{\sum_{\substack{k = 1\\k\neq i}}^M\: \frac{\psi_k}{N} \left(1+\frac{\varrho_k \varrho_i}{N-1}\right)\right\}\;,
\end{equation}
where $\mathds{E}\left[\Delta_i\right]$ is the expected degree of vertices induced by receiver $i$.
\end{theorem}
\begin{proof}
The proof can be found in \appref{app:graph-density}.
\end{proof}
\ignore{Using this expression in \thref{th:expected-edge-set-size}, we can introduce the expression for the edge set size evolution for any arbitrary chosen clique $\kappa$ with a targeted set $\mathcal{T}$ in the following theorem.}
\begin{theorem} \label{th:expected-edge-set-size-evolution}
For any given feedback state and any given maximal clique $\kappa$, chosen for transmission at time $t$, the expected edge set size of the IDNC graph at time $t+1$ is expressed as:
\begin{equation}\label{eq:expected-edge-set-size-evolution}
\mathds{E}\left[\left|\mathcal{E}^{(t+1)}\right|\right] = \;\; \mathds{E}\left[\left|\mathcal{E}^{(t)}\right|\right]- \frac{1}{2} \sum_{i\in\mathcal{T}} q_i \mathds{E}\left[\Delta_i^{(t)}\right] \ignore{ \nonumber\\
&} + \frac{1}{2}\sum_{i\in\mathcal{T}}  \psi_i \left(\alpha_i  - \frac{q_i \gamma_i}{\psi_i}\right) + \frac{1}{2}\sum_{i\notin\mathcal{T}}  \psi_i \beta_i\;,
\end{equation}
where
\begin{align}
&\alpha_i  = \sum_{\substack{k=1\\k\neq i}}^M q_i \xi_k -\sum_{\substack{k\in\mathcal{T}\\k\neq i}}\Phi_k(q_i) \ignore{\;, \label{eq:alpha}\\
&} \qquad\beta_i  = -\sum_{\substack{k\in\mathcal{T}\\k\neq i}}\Phi_k(0)\ignore{\;,\label{eq:beta}\\
&} \qquad\gamma_i  = \sum_{\substack{k=1\\k\neq i}}^M \xi_k -\sum_{\substack{k\in\mathcal{T}\\k\neq i}}\Phi_k(1) \ignore{\;,}
\label{eq:gamma} \\
&\Phi_k(x) = \frac{q_k}{N} \left( 1 + \frac{\left(\varrho_k-\psi_k+1\right) \left(\varrho_i+x\right)}{N-1}\right) \ignore{\;, \label{eq:phi}\\
&} \qquad \qquad~\xi_k = \frac{\psi_k \varrho_k}{N(N-1)}\ignore{\;.}\label{eq:xi}
\end{align}
\end{theorem}
\begin{proof}
The proof can be found in \appref{app:expected-edge-set-size-evolution}.
\end{proof}
The above formula in \eqref{eq:expected-edge-set-size-evolution} shows that the expected edge set size at $t+1$ is affected by two main components with respect to its value at time $t$. First, it suffers from a reduction due to the potential disappearance of served vertices and their edges, which is quantified by $\frac{1}{2} q_i \mathds{E}\left[\Delta_i^{(t)}\right]$ for each targeted receiver. The second component is the change in expected degrees of the remaining vertices, which is quantified by $\frac{1}{2} \psi_i\left(\alpha_i - \frac{q_i\gamma_i}{\psi_i}\right)$ or $\frac{1}{2} \psi_i\beta_i$ for each targeted or non-targeted receiver, respectively. We can thus investigate the coding strategy, maximizing the expected edge set size evolution, by studying these two components in the next section.

\section{Analysis of the Expected Coding Opportunity Evolution} \label{sec:receiver-selection-strategy}
To identify a strategy maximizing the expected number of coding opportunities, we will analyze two main components affecting the derived expression in \eqref{eq:expected-edge-set-size-evolution}.

\subsection{Vertex Disappearance}
The disappearance of vertices and its attributed loss of their adjacent edges is a natural outcome of targeting the packet requests of the different receivers throughout the recovery transmission process and is unavoidable. Nonetheless, we can still reduce the effect of this loss component by serving the vertices with smaller degrees. The following theorem compares the expected vertex degrees of two receivers given the sizes of their Wants sets.
\begin{theorem} \label{th:degree-comparison}
If $\psi_i > \psi_h$, then $\mathds{E}\left[\Delta_h\right] > \mathds{E}\left[\Delta_i\right]$.
\end{theorem}
\begin{proof}
The proof can be found in \appref{app:degree-comparison}.
\end{proof}
Now, if $q_i<q_h$ and $\psi_i>\psi_h$, $\frac{1}{2} q_i\mathds{E}\left[\Delta_i\right] < \frac{1}{2} q_h\mathds{E}\left[\Delta_h\right]$. Consequently, serving receivers with largest Wants sets and
erasure probabilities results in a smaller loss in the resulting edge set size.

\subsection{Degrees of Remaining Vertices}
To study the factors affecting the evolution of the degrees of remaining vertices, we introduce the following theorem.
\begin{theorem} \label{th:alpha-beta-gamma}
If $\psi_i > 0$, $\psi_k > 1,~\forall~k\in\mathcal{T}\setminus i$ and $\psi_k \leq \varrho_k,~\forall~k\in\mathcal{T}\setminus i$, then $\alpha_i - \frac{q_i\gamma_i}{\psi_i} \geq \beta_i$ for any $i\in\mathcal{M}$.
\ignore{\begin{equation}
\alpha_i - \frac{q_i\gamma_i}{\psi_i} \geq \beta_i\;.
\end{equation}}
\end{theorem}
\begin{proof}
The proof can be found in \appref{th:alpha-beta-gamma}.
\end{proof}
\ignore{The conditions in the above theorem\ignore{ generally} hold for most of the recovery transmission phase. Indeed, receivers with $\psi_k = 0$ are not in $\mathcal{T}$ in the first place. Receivers with $\psi_k = 1$ are on the edge of completing the frame reception and thus will be less targeted if receivers having larger $\psi_i$ exist. Finally, the condition $\varrho_k \geq \psi_k,~\forall~k\in\mathcal{T}\setminus i$ is usually true for erasure probabilities less than half. Even in the rare cases where this condition does not hold in the\ignore{ very} beginning of the recovery phase, it will definitely hold as the recovery phase progresses. We can thus conclude that the above theorem is valid for most of the recovery phase.}

The above theorem implies that\ignore{, if its condition applies (which is satisfied in most of the recovery phase), then} the increase in the degrees of the remaining vertices of any receiver is larger when it is targeted than when it is not.
Thus, moving a receiver $i$ from the non-targeted set to the targeted set results in adding $\frac{\psi_i}{2} \left(\alpha_i  - \frac{q_i\gamma_i}{\psi_i}  - \beta_i\right)\geq 0$ edges to $\mathcal{G}$. This term is larger when $\psi_i$ is larger, and thus moving a receiver with a larger Wants set to the targeted receiver set adds more edges to the primary graph than moving a receiver with a smaller Wants set. Consequently, a larger increase in the expected edge set size is obtained when targeting the maximum number of receivers having larger Wants sets.

Another important insight about the values of $\alpha_i$ and $\beta_i$ can be inferred from a closer look at their components $\Phi_k(q_i)$ and $\Phi_k(0)$, respectively. Since the terms $\sum_{\substack{k\in\mathcal{T}_\rho(\kappa)\\k\neq i}}\Phi_k(q_i)$ and $\sum_{\substack{k\in\mathcal{T}_\rho(\kappa)\\k\neq i}}\Phi_k(0)$ are subtractive terms from $\alpha_i$ and $\beta_i$, respectively, selecting the receivers, with smaller $\Phi_k(q_i)$ and $\Phi_k(0)$, as targeted receivers, increases the values of $\alpha_i$ and $\beta_i$, respectively. Now, if $q_k<q_h$, $\psi_k > \psi_h$, we have:
\begin{equation}
q_k\left(\varrho_k-\psi_k+1\right) < q_h\left(\varrho_h-\psi_h+1\right) \qquad \ignore{\\}
\Rightarrow \qquad \Phi_k(q_i) < \Phi_h(q_i) \quad\mbox{and}\quad \Phi_k(0) < \Phi_h(0) \;.
\end{equation}
Consequently, the receivers having larger Wants sets and erasure probabilities have smaller values of $\Phi_k(q_i)$ and $\Phi_k(0)$, and thus targeting them increases $\alpha_i$ and $\beta_i$. This result also makes the term $\frac{q_i\gamma_i}{\psi_i}$ negligible in the $\left(\alpha_i-\frac{q_i\gamma_i}{\psi_i}\right)$ term for targeted receivers having larger Wants sets and erasure probabilities.

From \thref{th:alpha-beta-gamma} and the above observations on the values of $\alpha_i$'s and $\beta_i$'s, we can conclude that targeting the maximum number of receivers with largest Wants sets and erasure probabilities results in a larger increase in the degrees of the remaining vertices in the IDNC graph.

\ignore{From these two introduced theorems and the above discussion, we can draw the following observations. First, the probable disappearance of served vertices will cause a smaller loss in the edge set size, if the maximum number of these vertices belongs to receivers with largest Wants sets\ignore{ and erasure probabilities}. Second, the increases in the expected degrees of the remaining vertices are larger when their inducing receivers are targeted than when they are not. Moreover, the values of these increases become larger when the receivers, with larger Wants sets and erasure probabilities, are targeted. Finally, since these changes in degrees are reflected on average in each of this receiver's vertices in the graph, which appears as multiplications of $\left(\alpha_i-\frac{q_i\gamma_i}{\psi_i}\right)$ and $\beta_i$ by $\psi_i$ in \eqref{eq:expected-edge-set-size-evolution}, a larger increase in the resulting expected edge set size is obtained when targeting the maximum number of receivers having largest Wants sets. Although targeting a larger number of receivers will in general cause more edge loss due to vertex disappearance, it decreases the denominator of the coding density expression and tends to increase the expected degrees of the remaining vertices. Thus, the overall coding density of the graph is expected to be increased.}

\subsection{Overall Strategy}
From the above two sections, we can conclude that both factors identified from \thref{th:edge-set-size-evolution} achieve a larger contribution in the number of edge set size at time $t+1$, with respect to its value at time $t$, when the sender targets the maximum number of receivers with largest Wants sets and erasure probabilities. We will refer to this strategy as the \emph{Worst Receiver Targeting (WoRT) strategy}. We can express this strategy as choosing the maximal clique $\kappa^*$ in each transmission such that:
\begin{equation} \label{eq:WoRT}
\kappa^* = \arg\min_{\kappa \in \mathcal{G}} \sum_{i|v_{ij}\in\kappa} \left(\frac{\psi_i}{q_i}\right)^n \;,
\end{equation}
where $n$ is a biasing factor.

\section{Simulations} \label{sec:simulations}
In this section, we test, through simulations, the performances of our identified strategies in maximizing the coding density of the system during the transmission of a frame, and compare them to other well-known strategies. The simulation scenario consists of $M$ receivers having different packet erasure probabilities\ignore{ ranging from $0.7$ to $0.99$,} while maintaining the average erasure probability ($\varepsilon$) constant. These erasure probabilities are assumed to be fixed during the transmission of a frame but change from frame to frame during the simulation. The tested strategies in this simulations are:
\begin{itemize}
\item RND: Random clique selection \cite{4476183}.
\item MC: Maximum clique selection \cite{4660042}.
\item MWC-R: Maximum weighted clique (MWC) selection, in which the weight of vertex $v_{ij}$ is defined as the reception probability of receiver $i$.
\item MoWPS: Maximum weight clique selection defined as in \eqref{eq:MoWPS}.
\item WoRT: Maximum weight clique selection defined as in \eqref{eq:WoRT}.
\end{itemize}
All figures represent the average coding density of the graph after each transmission, when the corresponding strategy is employed during all the recovery transmission phase. This average is computed over a large number of iterations, in each of which we compute the graph density after each transmission. We then average all the densities evaluated at the same transmission index.

\subsection{Results with Optimal MWC Algorithm}
\fref{fig:CE-N20-VarGrph} depicts the coding density evolution inside the IDNC graph for $M = 50$ and $N=20$. The erasure probabilities of different receivers take values in the range of 0.01 to 0.3, such that $\varepsilon =0.15$.
\begin{figure}[t]
\centering
  \includegraphics[width=0.75\linewidth]{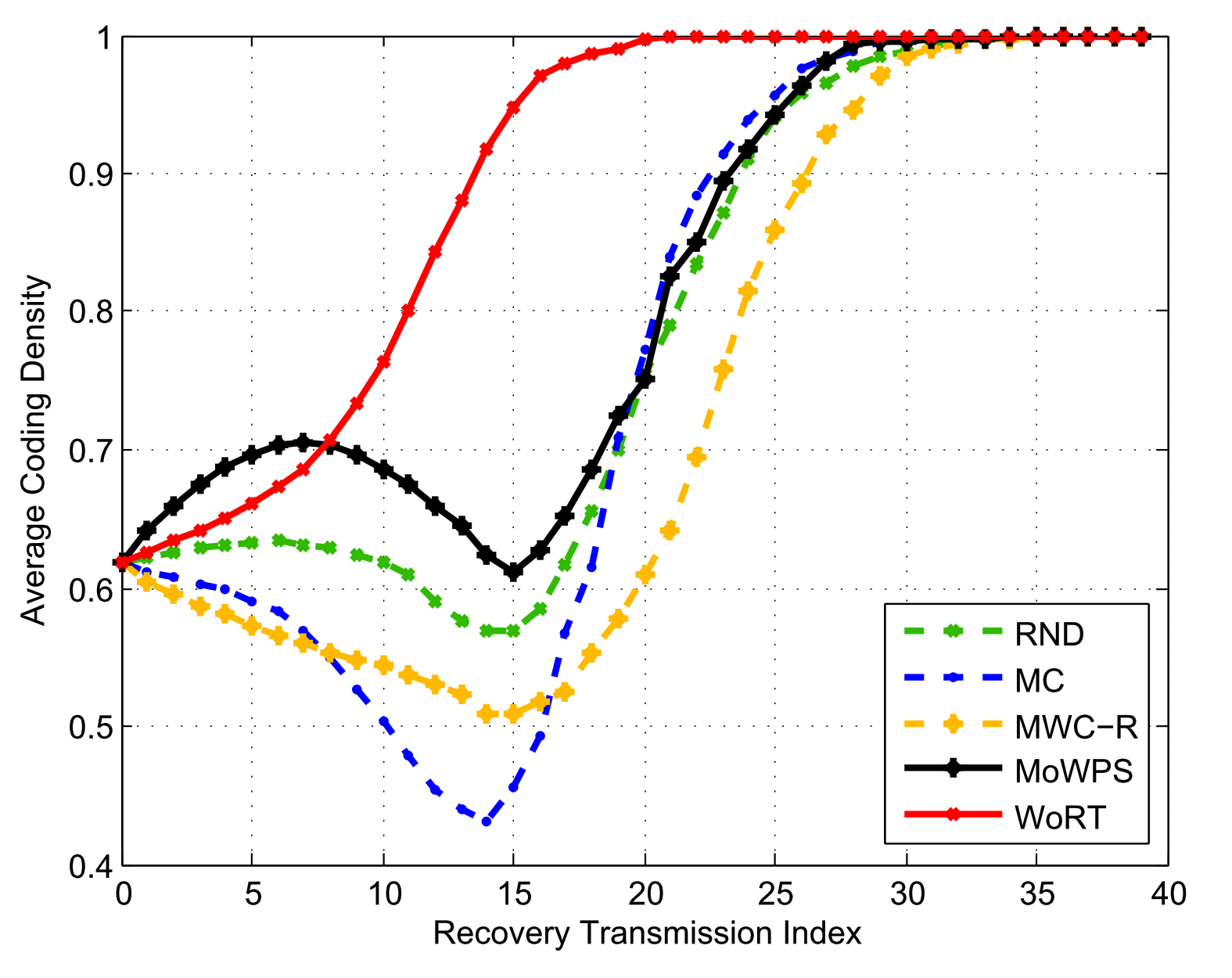}\\
  \caption{Average coding density evolution for $M=50$ and $N=20$}\label{fig:CE-N20-VarGrph}
\end{figure}
From \fref{fig:CE-N20-VarGrph}, we can draw the following observations. As expected, the MoWPS strategy considerably increases the coding density for the first 25$\%$ of the recovery transmissions (i.e. transmissions 1 to 8) due to the presence of packets requested by a large number of receivers during this period. However, due to the reduction of these numbers after several transmissions, the strategy cannot continue increasing the coding opportunities. Indeed, the restriction of serving such packets when there are none forces the algorithm to serve less vertices, which results in a smaller coding density, as shown in the intermediate 22$\%$ of the recovery phase (i.e. transmissions 9 to 15). Towards the end of the phase, the number of requests is naturally reduced and the Has sets become very close to $N$.\ignore{ From \thref{th:degree-comparison}, the number of vertices with larger degrees thus increases, which in turn increases the clique sizes and thus number of targeted receivers in each transmission.} This naturally increases the coding opportunities due to Condition C2 in \sref{sec:IDNC-graph}. Consequently, the remaining transmissions target most of the receivers, which both further increases the coding opportunities and reduces the number of vertices in the graph, thus increasing the coding density.

As for the effect of receiver selection, we can clearly see that the WoRT strategy considerably outperforms all other receiver selection strategies. Moreover, we can see that it monotonically increases the coding density, which proves its ability to increase the coding density over the full course of the frame transmission. This performance is supported by the fact that a single transmission can reduce the Wants set of any receiver by at most one. Consequently, there will always exist some receivers with larger Wants set sizes during most of the recovery phase, which can be targeted to continuously increase the number of coding opportunities and the coding density.

Another important result is that the MC strategy, serving the maximum number of receivers (or requests) in each transmission, and widely considered in most opportunistic network coding works, results in a very bad evolution of coding density. We can infer from \fref{fig:CE-N20-VarGrph} that this strategy serves very few large cliques in the very beginning of the recovery phase and then is left with smaller cliques to serve, until it is close to the completion of frame delivery. This can be explained from \eqref{eq:degree-comparison}, showing the high adjacency of vertices with smaller Wants sets. Consequently, the MC strategy mostly targets receivers with smaller Wants sets, which clearly reduces the edge set size according to \sref{sec:receiver-selection-strategy}.\ignore{However, towards the end of the recovery phase, all techniques catch up with their efficiency similar to MOWPS.} Despite the common intuition that this strategy can optimize different network coding parameters, this result shows that it may not be able to truly do so. We will show one example of this fact in the next section.

Figures \ref{CD-M} and \ref{CD-N} depict the coding density evolution inside the IDNC graph for different values of $M$ (when $N=20$ and $\varepsilon =0.15$) and $N$ (when $M=50$ and $\varepsilon =0.15$), respectively.

\begin{figure}[p]
  \centering
  \includegraphics[width=0.75\linewidth]{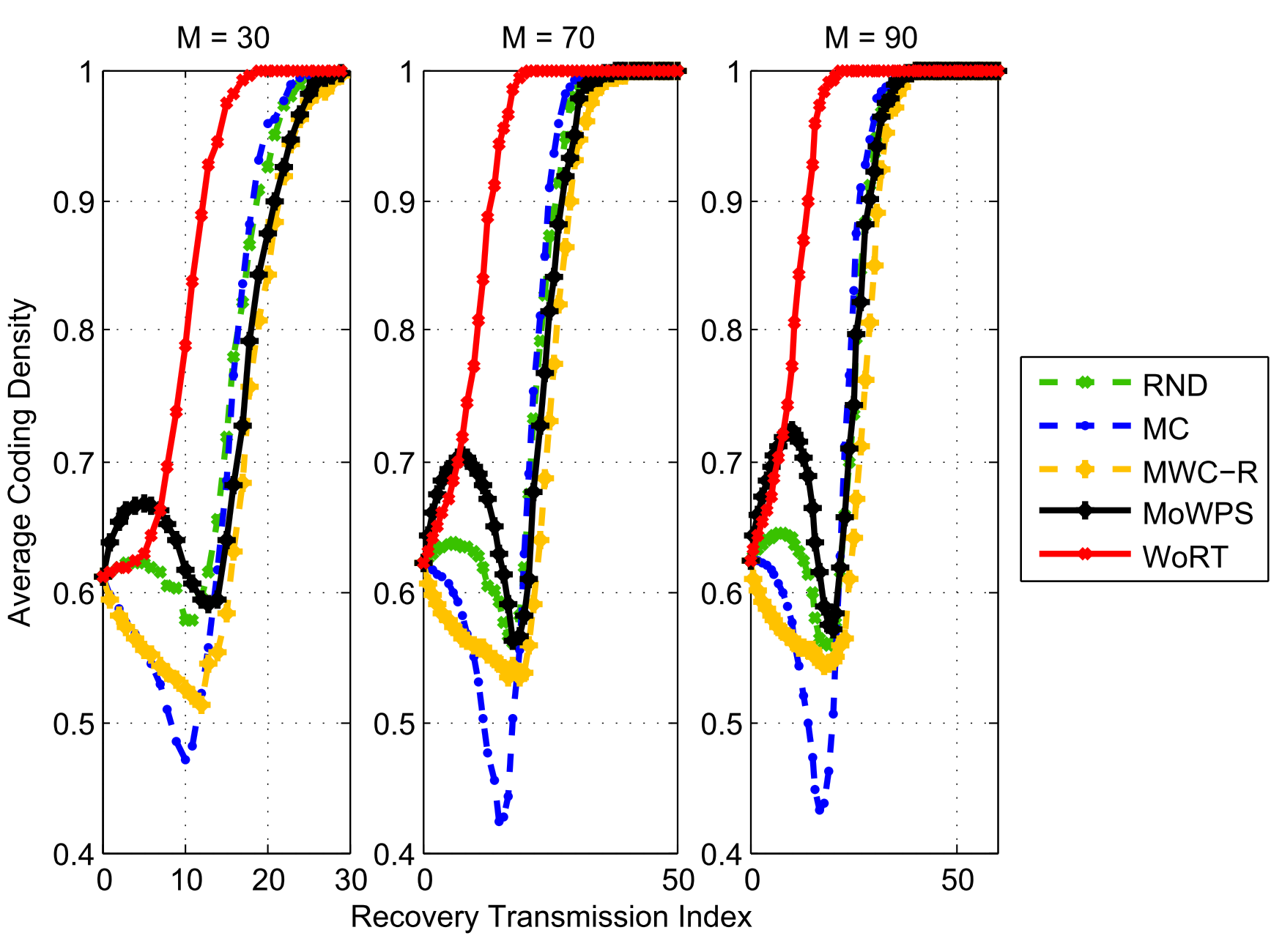}\\
  \caption{Average coding density evolution for different values of $M$, when $N=20$ and $\varepsilon=0.15$}\label{CD-M}
\end{figure}

\begin{figure}
  \centering
  \includegraphics[width=0.75\linewidth]{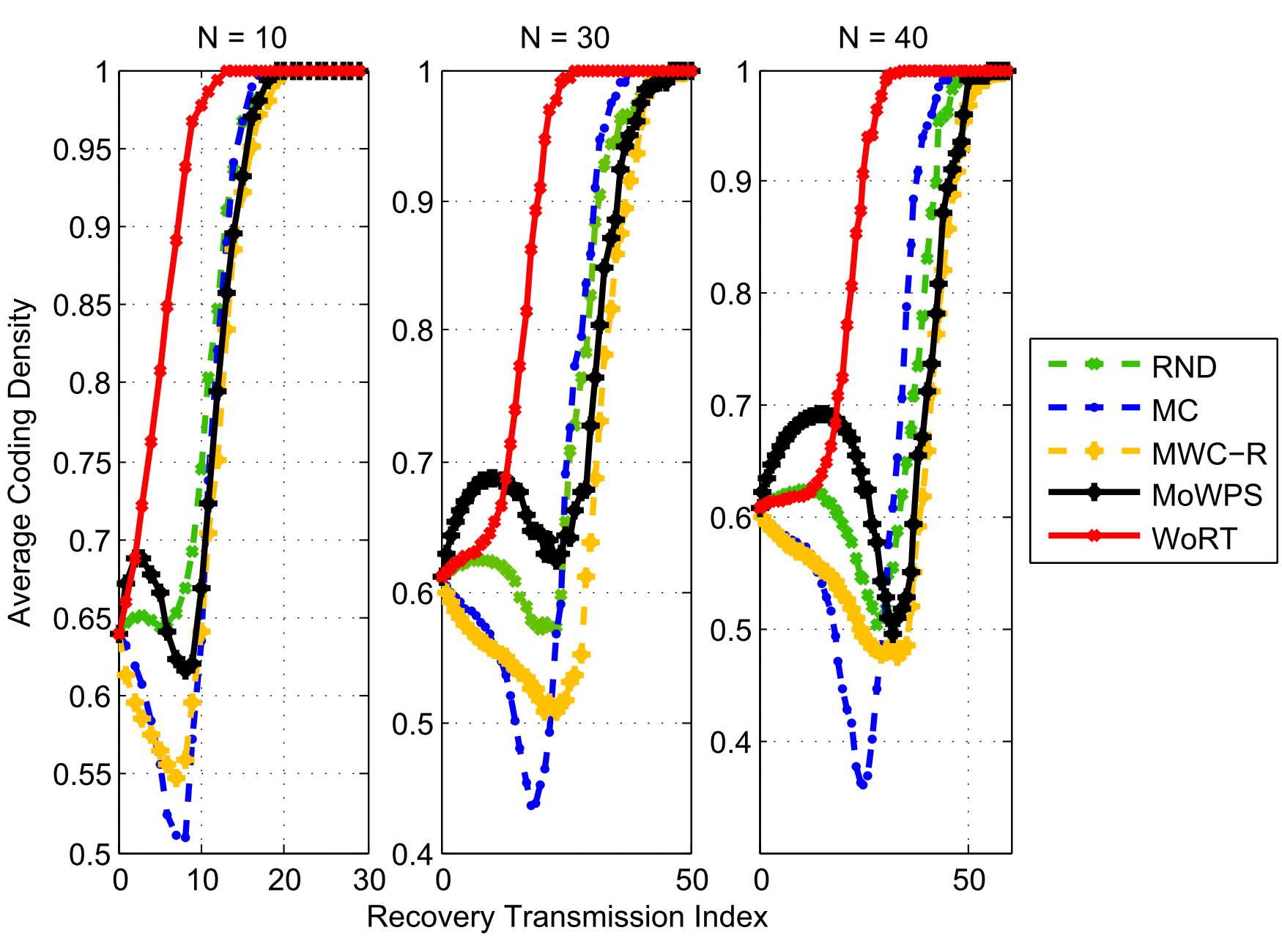}\\
  \caption{Average coding density evolution for different values of $N$, when $M=50$ and $\varepsilon=0.15$}\label{CD-N}
\end{figure}
We can see from both figures that all the observation and conclusions deduced from \fref{fig:CE-N20-VarGrph} hold for various values of $M$ and $N$.

\subsection{Results with Greedy MWC Algorithm}
Since the optimal solution of MWC selection problem is NP-hard, we test the performance of a greedy MWC algorithm, which adds the vertex with the highest weight to the output clique in each iteration. For a better representation of the adjacency effect on the vertex selection in this greedy approach, we modify the weight of each vertex $v_{ij}$, having original weight $w_{ij}$, to be:
\begin{equation}
\omega_{ij} = w_{ij} \cdot \sum_{\forall~v_{kl}\in\mathcal{V}} \;I_{v_{kl}\in \mathcal{V}_{ij}}\cdot w_{kl}
\end{equation}
where $\mathcal{V}_{ij}$ is the set of adjacent vertices to $v_{ij}$ and where $I_x$ is an indicator function, which is equal to one if $x$ is true and zero otherwise. Consequently, these new vertex weights reflect not only their individual weights $w_{ij}$ but also their adjacency to a large number of vertices having high individual weights. Since these modified weights need to be recomputed after each iteration to reflect the new adjacency conditions, the complexity of this greedy algorithms is $O(M^2N)$.

\fref{fig:CE-N20-srh} depicts the coding density evolution of optimal and greedy MWC algorithms for $M = 50$, $N=20$ and $\varepsilon=0.15$. In this figure, we only consider the MoWPS and WoRT strategies.
\begin{figure}
\centering
  \includegraphics[width=0.75\linewidth]{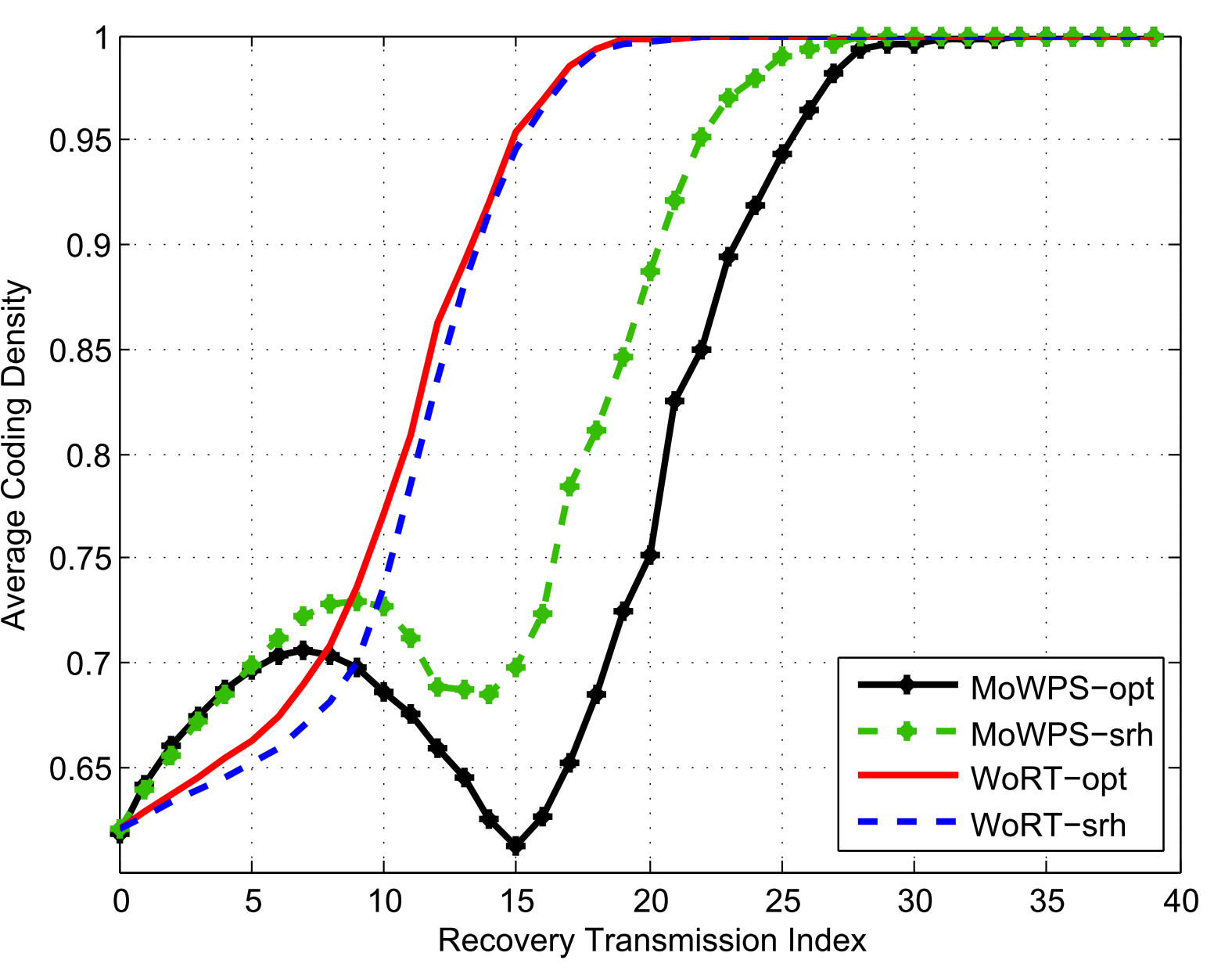}\\
  \caption{Average coding density evolution for $M=50$ and $N=20$.}\label{fig:CE-N20-srh}
\end{figure}
For the WoRT strategy, we can see that the greedy algorithm achieves the same monotonically increasing performance as the optimal algorithm, with a slight gap between them\ignore{ in the mid-range of the recovery transmissions}. For the MoWPS strategy, we notice that the greedy algorithm follows the same trend of the optimal algorithms but achieves larger coding density after the first 7 transmissions. This can be explained by the fact that the greedy algorithm loses some chances of serving packets with largest demands in the first few transmissions, compared to the optimal algorithm. This luckily makes it less affected by the degradation phenomenon that happens to the optimal algorithm, as explained above. This effect both preserves better coding opportunities for the greedy algorithm in the intermediate range of transmissions and allows a faster boost up towards the end.

\subsection{Effect of Erasure Probabilities}
Since the receiver selection in the WoRT strategy greatly depends on their erasure probabilities, we test its performance for several erasure probabilities $\varepsilon_w = [0.3, 0.5, 0.7, 0.9]$ in \fref{fig:erasure-test}.
\begin{figure}
\centering
  \includegraphics[width=0.75\linewidth]{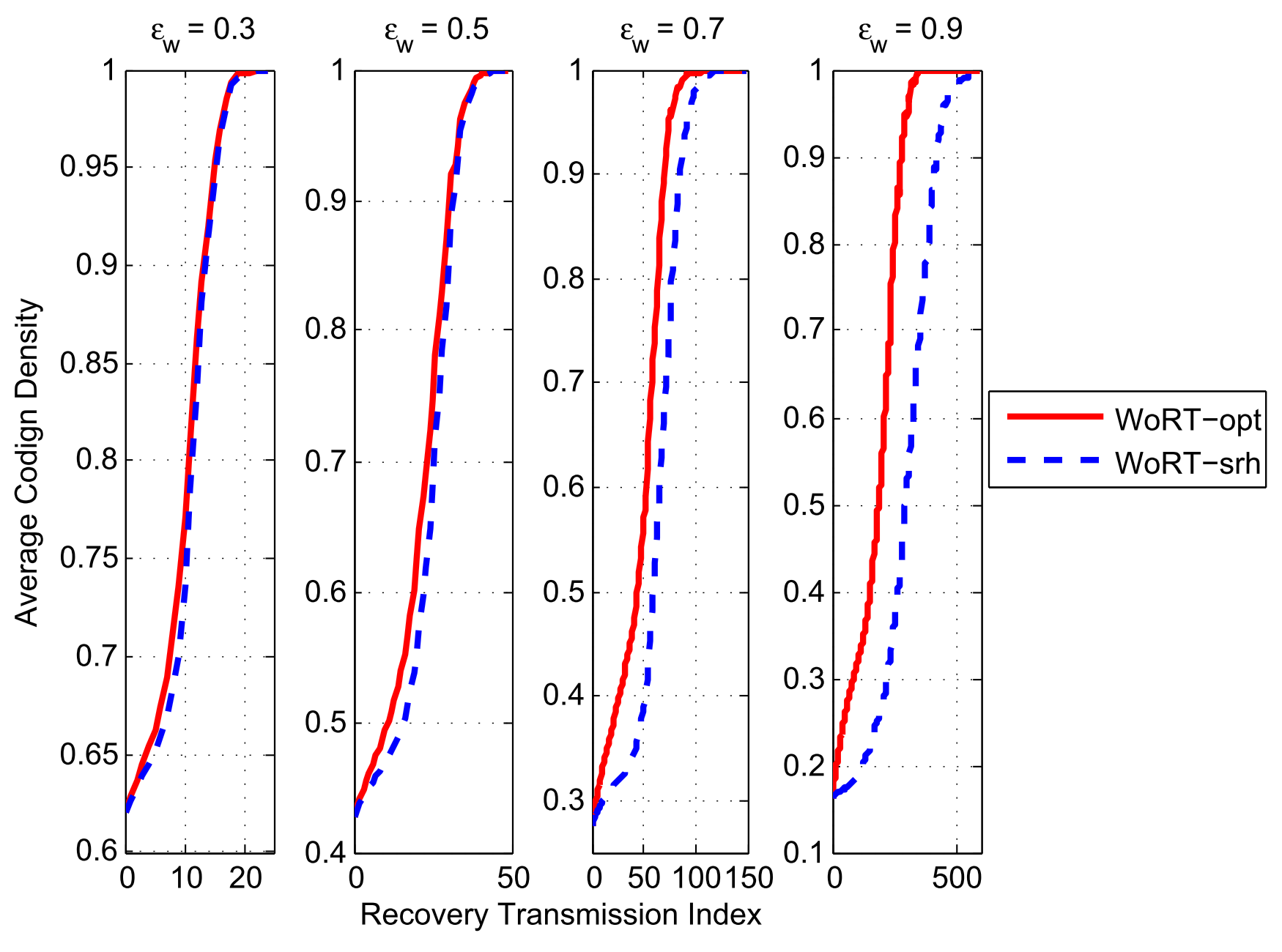}\\
  \caption{Average coding density evolution for different $\varepsilon_w$ ($M=50$ and $N=20$).}\label{fig:erasure-test}
\end{figure}
We can see from the figure that the performances of both the optimal and greedy MWC algorithms always achieve the same monotonically increasing trend of average coding densities for all these diverse values of erasure probabilities. We also notice a larger delay between the optimal and greedy algorithms as the erasure probability increases. This result is expected as the greedy algorithm will naturally degrade more in the overall performance as it runs for a larger number of recovery transmissions at high erasure probabilities.

\section{Case Studies}  \label{sec:case-study}
\subsubsection{Completion Delay}
In this section, we study the effect of maximizing the coding density on reducing the average completion delay (i.e. the number of recovery transmissions) in IDNC. Intuitively, one can think that the optimal IDNC completion delay can be achieved by serving the maximum number of vertices in each recovery transmission (i.e. the MC strategy), as this should apparently deplete the graph faster. However, we have shown that the MC strategy suffers from severe degradation in the coding density compared to the WoRT strategy. We will thus compare the completion delay of these two strategies to see whether increasing the coding density in earlier transmissions is important in reducing the frame completion delay.

\fref{fig:CD_Comparison} depicts the comparison of the average completion delays achieved by the RND, MC, MoWPS and WoRT algorithms to global optimal completion delay over all linear network codes. The upper subfigure illustrates this comparison against $M$, for $N = 40$ and $\varepsilon=0.15$, whereas the lower one compares the performances against $N$ for $M=40$ and $\varepsilon=0.15$. Both subfigures show that the WoRT strategy considerably outperforms all other strategies including the MC strategy, especially as $M$ and $N$ increase. We can also see that WoRT strategy achieves a near-optimal completion delay performance.
\begin{figure}
\centering
  \includegraphics[width=0.75\linewidth]{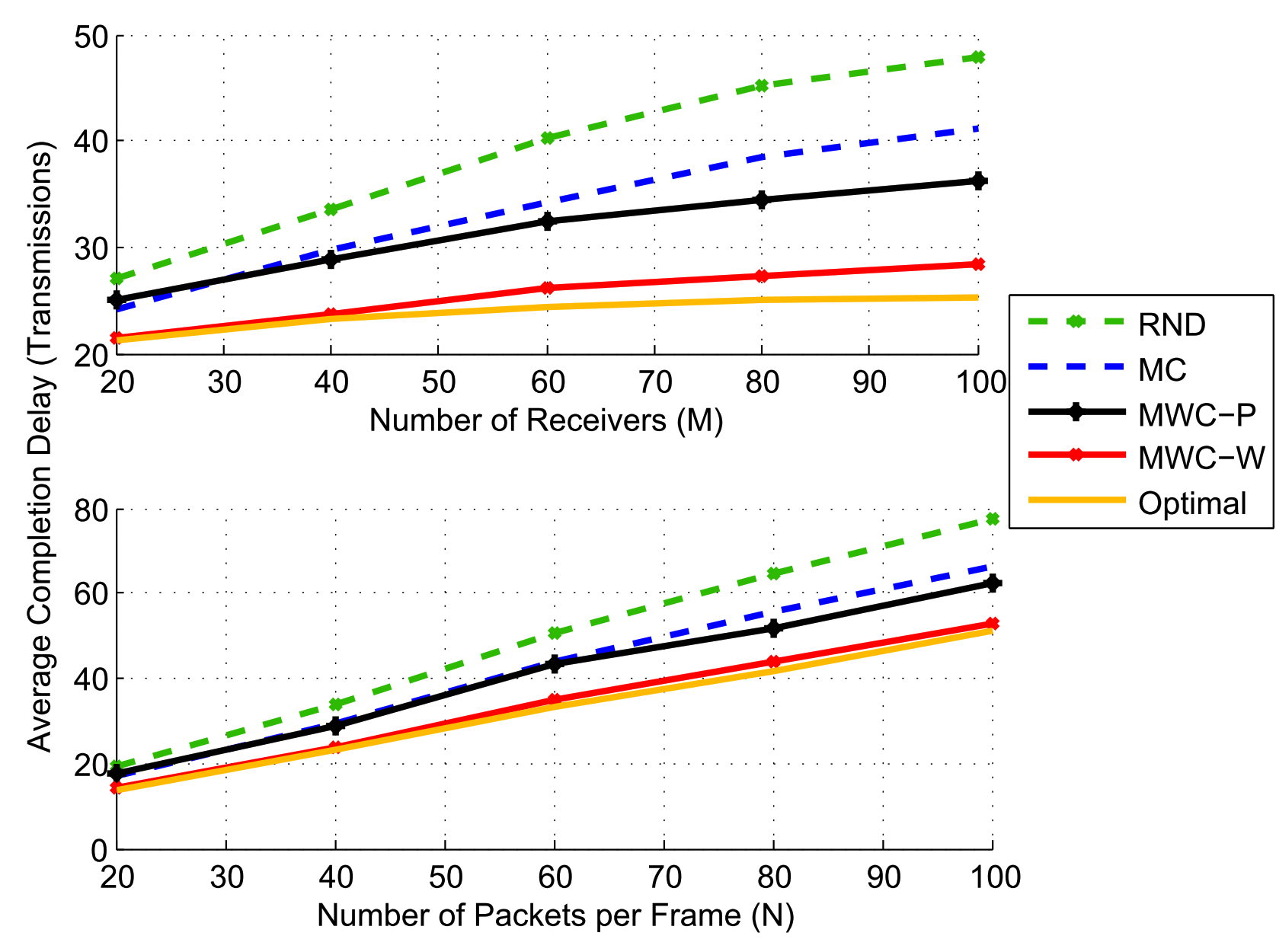}\\
  \caption{Comparison of average completion delays against $M$ and $N$.}\label{fig:CD_Comparison}
\end{figure}
These results clearly show that the MC strategy cannot achieve a low completion delay, due to the effect explained in \sref{sec:simulations}. On the other hand, the WoRT strategy, which serves relatively smaller cliques in the beginning of the recovery phase, achieves a better completion delay as it persistently increases\ignore{ the adjacency of the vertices belonging to the receivers with largest Wants sets and erasure probabilities} the coding density in the graph. Thus, it always finds large cliques to serve in later recovery transmissions, and thus complete their delivery much faster.

\subsection{Receiver Goodput}
We formally define the receiver goodput as the ratio between the number of useful packets (i.e. new source packets or instantly decodable recovery packets) received by a given receiver to the total number of received packets by this receiver, along a large number of frames. This definition does not count channel inflicted useless transmissions due to erasures, but
rather counts sender inflicted useless transmissions, when its coding algorithm is not able to provide instantly-decodable packets to different receivers. In other words, the receiver goodput can be viewed as an erasure independent received throughput.

\fref{fig:OGP_Comparison} depicts the comparison of the average receiver goodput achieved by the RND, MC, MoWPS and WoRT algorithms. The upper and lower subfigures illustrates this comparison for the same parameters of the upper and lower subfigures of \fref{fig:CD_Comparison}
\begin{figure}
\centering
  \includegraphics[width=0.75\linewidth]{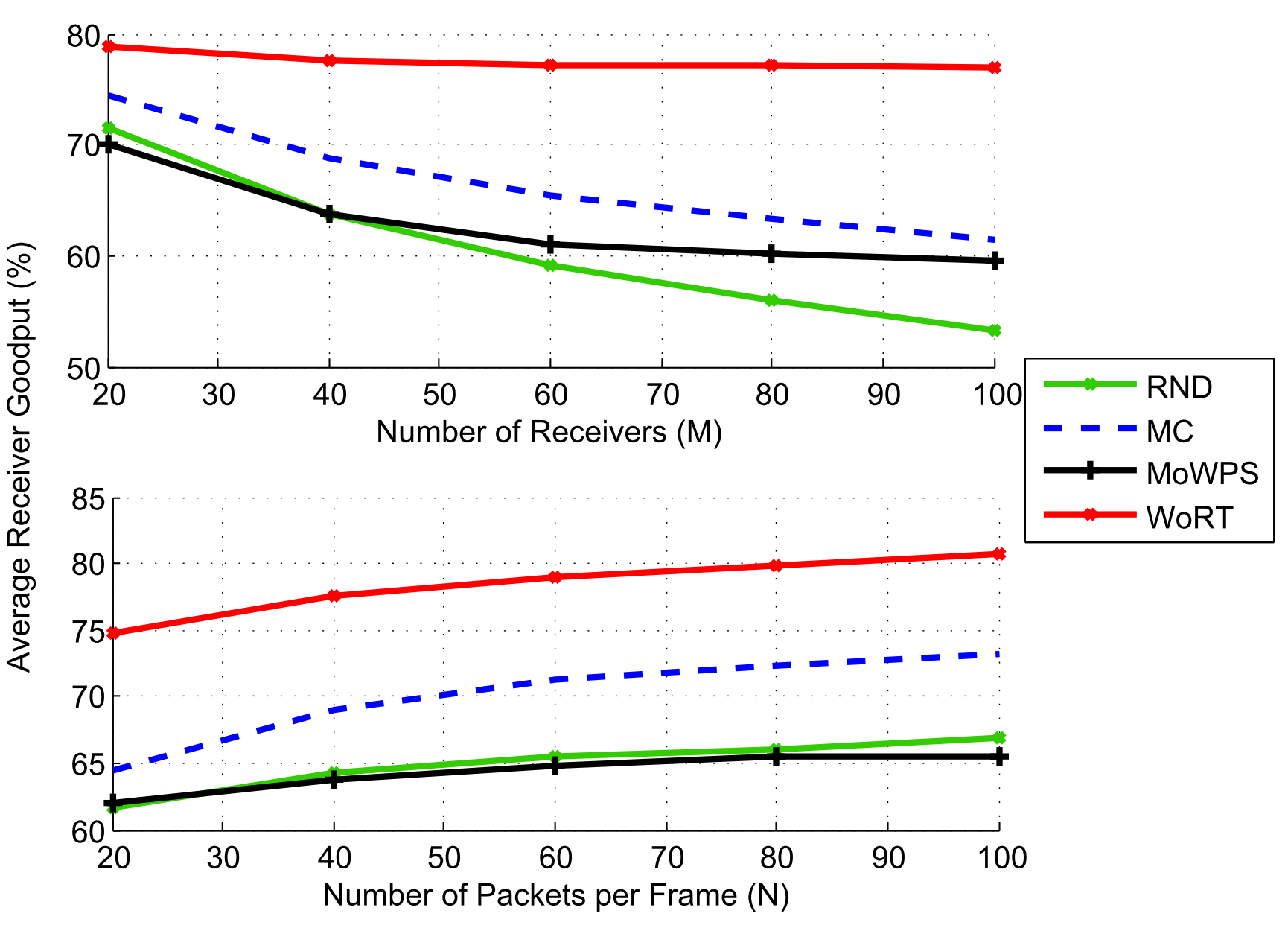}\\
  \caption{Comparison of average receiver goodputs against $M$ and $N$.}\label{fig:OGP_Comparison}
\end{figure}
Again the results show that the WoRT strategy achieves the best goodput performance, due its ability to persistently increase high coding density in the IDNC graph, which allows to continuously provide useful packets to a large number of receivers all along the recovery transmission phase. All other schemes lag behind because they are less efficient in maintaining high coding density, and thus end up sending transmissions that are useful to a limited number of receivers.

\section{Conclusion} \label{sec:conclusion}
In this paper, we investigated the receiver and packet selection strategies that densify the IDNC coding opportunities in wireless broadcast. We first derived an expression for the exact evolution of the edge set size in the IDNC graph, after the transmission of any arbitrary coded packet. From this expression, we showed that sending commonly wanted packets for all the receivers can maximize the number of coding opportunities. Since guaranteeing such property in IDNC is usually impossible, especially after the first few transmission, this strategy does not guarantee the continuous increase of the coding density. This observation has been later demonstrated through extensive simulations. Consequently, we further investigate the problem by deriving an expected expression of the edge set size evolution, after ignoring the feedback set contents and keeping their cardinalities. We then employed this expression to show that targeting the maximum number of receivers having the largest Wants sets and erasure probabilities tends to both maximize the expected number of coding opportunities and monotonically increase the expected coding density. Simulation results justified our theoretical findings. Finally, we demonstrated the validity and importance of increasing the coding density through two case studies on the IDNC completion delay and receiver goodput. The studies showed that the identified WoRT strategy achieves a significantly better performance than the MC strategy (which was intuitively expected to perform better), due to the inability of the latter to preserve high coding density throughout the frame transmission. We thus recommend the observance of the coding density increasing strategies when other network parameters are optimized over the full course of a frame transmission.

\appendices

\section{Proof of \thref{th:edge-set-size}} \label{app:edge-set-size}
It is well known from graph theory that the edge set size of any graph is equal to half the sum of its vertex degrees. Consequently, we will first derive the expression for the degree of an arbitrary vertex $v_{ij}$ in the graph. From the adjacency conditions C1 and C2 in \sref{sec:IDNC-graph}, we conclude the following facts:
\begin{itemize}
\item Vertex $v_{ij}$ is not adjacent to any vertex of the same receive $i$.
\item If $j\in\mathcal{W}_k$, $v_{ij}$ cannot be adjacent to any vertex of receiver $k$ due to violation of C2, except for vertex $v_{kj}$ which arises from C1.
\item If $j\in\mathcal{H}_k$, $v_{ij}$ can be adjacent to any vertex of receiver $k$ (induced from $\mathcal{W}_k$), except for all vertices $v_{kl}$ for which $l\notin\mathcal{H}_i \Rightarrow l\in\mathcal{W}_i$.
\end{itemize}
From these facts, we can express the degree of a vertex $v_{ij}$ as follows:
\begin{equation} \label{eq:vertex-degree}
\Delta_{ij} = \sum_{\substack{k=1\\k\neq i}}^M I_{j\in\mathcal{W}_k} + I_{j\in\mathcal{H}_k}\left(\left|\mathcal{W}_k\right| - \left|\mathcal{W}_k \cap \mathcal{W}_i\right|\right)\;
\end{equation}
Now, the sum $\Sigma \Delta_i$ of all the degrees of the vertices induced by receiver $i$ can be expressed as:
\begin{align} \label{eq:sigma-delta}
\Sigma \Delta_i &= \sum_{j\in\mathcal{W}_i} \sum_{\substack{k=1\\k\neq i}}^M I_{j\in\mathcal{W}_k} + I_{j\in\mathcal{H}_k}\left(\left|\mathcal{W}_k\right| - \left|\mathcal{W}_k \cap \mathcal{W}_i\right|\right) \ignore{\nonumber \\
&} = \sum_{\substack{k=1\\k\neq i}}^M \left(\sum_{j\in\mathcal{W}_i}I_{j\in\mathcal{W}_k}\right) + \sum_{j\in\mathcal{W}_i}I_{j\in\mathcal{H}_k}\cdot \left(\psi_k - \left|\mathcal{W}_k \cap \mathcal{W}_i\right|\right) \nonumber\\
& = \sum_{\substack{k=1\\k\neq i}}^M \left|\mathcal{W}_i \cap \mathcal{W}_k\right|  + \sum_{j\in\mathcal{W}_i}I_{j\in\mathcal{H}_k}\cdot \left(\psi_k - \left|\mathcal{W}_k \cap \mathcal{W}_i\right|\right) \;.
\end{align}
The expression $\sum_{j\in\mathcal{W}_i}I_{j\in\mathcal{H}_k}$ can be easily shown to be equal to $\left|\mathcal{W}_i\cap\mathcal{H}_k\right|$. We can also easily infer that $\mathcal{W}_i\cap\mathcal{H}_k = \mathcal{W}_i\setminus\left(\mathcal{W}_i\cap\mathcal{W}_k\right)$. Consequently, we get:
\begin{equation} \label{eq:set-conversion}
\sum_{j\in\mathcal{W}_i}I_{j\in\mathcal{H}_k} = \left|\mathcal{W}_i\cap\mathcal{H}_i\right| = \left|\mathcal{W}_i\right|-\left|\mathcal{W}_i\cap\mathcal{W}_k\right|\;.
\end{equation}
From \eqref{eq:set-conversion}, \eqref{eq:sigma-delta} and using the definitions in \eqref{eq:edge-set-size}, we get:
\begin{equation}
\left|\mathcal{E}\right| = \frac{1}{2}\sum_{i=1}^M \Sigma \Delta_i \ignore{\nonumber\\
&} = \frac{1}{2}\sum_{i=1}^M\sum_{\substack{k=1\\k\neq i}}^M \left[\psi_{ik}  + \left(\psi_i - \psi_{ik}\right)\left(\psi_k - \psi_{ik}\right)\right]\ignore{\nonumber\\
&} = \frac{1}{2}\sum_{i=1}^M\sum_{\substack{k=1\\k\neq i}}^M \left(\psi_{ik} + \theta_{ik}\theta_{ki}\right)\;.
\end{equation}

\section{Proof of \thref{th:edge-set-size-evolution}} \label{app:edge-set-size-evolution}
When an arbitrary clique $\kappa$, with a given set of targeted receivers $\mathcal{T}$, is chosen for transmission, the values of $\psi_i$, $\psi_k$ and $\psi_{ik}$ in \eqref{eq:edge-set-size} change according to the cases depicted in \tref{tab:evolution-conditions} for each pair of receivers. In this table, $X_{ik}$ is a joint indicator function, which is equal to 1 if either $i$ or $k$ received the transmitted packet and is zero otherwise. It is easy to see that $X_{ik}$ can be expressed as $X_{ik} = X_i + X_k - X_i X_k$.
\begin{table}
\centering
\caption{Evolution cases of the Wants sets of a pair of receivers $i$ and $k$ when the sender transmits an arbitrary clique $\kappa$ targeting a set of receivers $\mathcal{T}$}\label{tab:evolution-conditions}
\begin{tabular}{|c|c|c|c|c|c|}
  \hline
  \multicolumn{3}{|c|}{Cases} & \multicolumn{1}{|c|}{\multirow{2}{*}{$\psi_i$}} & \multicolumn{1}{|c|}{\multirow{2}{*}{$\psi_k$}} & \multicolumn{1}{|c|}{\multirow{2}{*}{$\psi_{ik}$}} \\ \cline{1-3}
  $i$ & $k$ & $p_i/p_k$ & \multicolumn{1}{|c|}{} & \multicolumn{1}{|c|}{} & \multicolumn{1}{|c|}{} \\ \hline \hline
  $\notin\mathcal{T}$ & $\notin\mathcal{T}$ & - & $\psi_i$ & $\psi_k$ & $\psi_{ik}$ \\ \hline
  $\notin\mathcal{T}$ & $\in\mathcal{T}$ & $p_k\notin\mathcal{W}_i$ & $\psi_i$ & $\psi_k - X_k$ & $\psi_{ik}$ \\ \hline
  $\notin\mathcal{T}$ & $\in\mathcal{T}$ & $p_k\in\mathcal{W}_i$ & $\psi_i$ & $\psi_k - X_k$ & $\psi_{ik} - X_k$ \\ \hline
  $\in\mathcal{T}$ & $\notin\mathcal{T}$ & $p_i\notin\mathcal{W}_k$ & $\psi_i - X_i$ & $\psi_k$ & $\psi_{ik}$ \\ \hline
  $\in\mathcal{T}$ & $\notin\mathcal{T}$ & $p_i\in\mathcal{W}_k$ & $\psi_i - X_i$ & $\psi_k$ & $\psi_{ik} - X_i$ \\ \hline
  $\in\mathcal{T}$ & $\in\mathcal{T}$ & $p_i\notin\mathcal{W}_k$ & $\psi_i - X_i$ & $\psi_k - X_k$ & $\psi_{ik}$ \\ \hline
  $\in\mathcal{T}$ & $\in\mathcal{T}$ & $p_i\in\mathcal{W}_k$ & $\psi_i - X_i$ & $\psi_k - X_k$ & $\psi_{ik} - X_{ik}$ \\ \hline
\end{tabular}
\end{table}
Let $Y_{ik}^{(t)} = \psi_{ik} + \theta_{ik}\theta_{ki}$. According to the possible changes in \tref{tab:evolution-conditions}, the double summation in \eqref{eq:edge-set-size}, can be divided into these seven categories. In each category, the values of $\psi_i$, $\psi_k$ and $\psi_{ik}$ in \eqref{eq:edge-set-size} are replaced in their $Y_{ik}^{(t)}$ expression by their corresponding evolved values, depicted in the table, to reach to $Y_{ik}^{(t+1)}$. Applying these changes, using the definition $X_{ik} = X_i + X_k - X_i X_k$, and re-arranging the terms, we get the following results for each category.
\begin{enumerate}
\item For $i\notin\mathcal{T}$, $k\notin\mathcal{T}$, there will be no change in $Y_{ik}^{(t)}$.
\item For $i\notin\mathcal{T}$, $k\in\mathcal{T}$ and $p_k\notin\mathcal{W}_i$, $Y_{ik}^{(t+1)} = \psi_{ik} + \theta_{ik} \left(\theta_{ki} - X_k\right) = Y_{ik}^{(t)} - X_k\theta_{ik}$.
\ignore{\begin{equation}
Y_{ik}^{(t+1)} = Y_{ik}^{(t)} - X_k\theta_{ik}\;.
\end{equation}}
\item For $i\notin\mathcal{T}$, $k\in\mathcal{T}$ and $p_k\in\mathcal{W}_i$, $Y_{ik}^{(t+1)} = \psi_{ik} - X_k + \left(\theta_{ik} + X_k\right) \theta_{ki} = Y_{ik}^{(t)} + X_k\hat{\theta}_{ki}$.
\ignore{\begin{equation}
Y_{ik}^{(t+1)} = Y_{ik}^{(t)} + X_k\left(\theta_{ki} - 1\right)\;.
\end{equation}}
\item For $i\in\mathcal{T}$, $k\notin\mathcal{T}$, we get the same results in the points 2 and 3, for $p_i\notin\mathcal{W}_k$ and $p_i\in\mathcal{W}_k$, respectively, by replacing index $k$ by $i$ and vice versa in the terms $X_k$, $\theta_{ik}$ and $\hat{\theta}_{ki}$.\ignore{ , we get:
\begin{equation}
Y_{ik}^{(t+1)} = Y_{ik}^{(t)} - X_i\theta_{ki}\;.
\end{equation}
\item For $i\in\mathcal{T}$, $k\notin\mathcal{T}$ and $p_k\in\mathcal{W}_i$, we get:
\begin{equation}
Y_{ik}^{(t+1)} = Y_{ik}^{(t)} - X_i\left(\theta_{ik} -1\right)\;.
\end{equation}}
\item For $i\in\mathcal{T}$, $k\in\mathcal{T}$ and $p_i\notin\mathcal{W}_k$, $Y_{ik}^{(t+1)} = \psi_{ik} + \left(\theta_{ik} - X_i\right) \left(\theta_{ki} - X_k\right) = Y_{ik}^{(t)} + X_iX_k - X_i\theta_{ki} - X_k\theta_{ik}$.
\ignore{\begin{equation}
Y_{ik}^{(t+1)} = Y_{ik}^{(t)} + X_iX_k - X_i\theta_{ki} - X_k\theta_{ik}\;.
\end{equation}}
\item For $i\in\mathcal{T}$, $k\in\mathcal{T}$ and $p_i\in\mathcal{W}_k$,
\begin{align*}
 Y_{ik}^{(t+1)} & = \psi_{ik} - X_i - X_k + X_iX_k  + \left(\theta_{ik} + X_k - X_iX_k\right) \left(\theta_{ki} + X_i - X_iX_k\right)\\
 & = Y_{ik}^{(t)} + X_i\hat{\theta}_{ik} + X_k\hat{\theta}_{ki}- X_i X_k \left(\hat{\theta}_{ik} + \hat{\theta}_{ki} - X_i X_k\right)\;.
\end{align*}
\end{enumerate}
The theorem follows by substituting the above equations in \eqref{eq:edge-set-size}.\ignore{, and knowing that the summations of $Y_{ik}{(t+1)}$ and $Y_{ik}{(t)}$ over all cases and all receiver pairs are equal to $\left|\mathcal{E}^{(t+1)}\right|$ and $\left|\mathcal{E}^{(t)}\right|$, respectively}

\section{Proof of \thref{th:two-receivers}} \label{app:two-receivers}
According to the targeting status of any two receivers $i$ and $k$, they will have a pairwise entry in only one of the summations in \eqref{eq:edge-set-size-evolution}. If $i\notin\mathcal{T}$ and $k\in\mathcal{T}$, the number of pairwise edges at $t+1$ $\left(Y_{ik}^{(t+1)}\right)$ is increased by $X_k\hat{\theta}_{ki}$ when $p_k\in\mathcal{W}_i$ and is reduced by $X_k\theta_{ik}$ when $p_k\notin\mathcal{W}_i$. If $i\in\mathcal{T}$ and $k\notin\mathcal{T}$, $Y_{ik}^{(t+1)}$ is increased by $X_i\hat{\theta}_{ik}$ when $p_i\in\mathcal{W}_k$ and is reduced by $X_i\theta_{ki}$ when $p_i\notin\mathcal{W}_k$.\ignore{
Consequently, serving a packet $p_k\in\mathcal{W}_k$, which is also in $\mathcal{W}_i$, increases the edge set size, whereas it reduces it when $p_k\notin\mathcal{W}_i$. Thus, for all non-targeted receivers, the number of pairwise edges with targeted receivers is maximized when the served packets are in their Wants sets.}

If both receivers are targeted, we have one of the following cases:\\
\textbf{Case 1}: $X_i = 0$ and $X_k=0$: $Y_{ik}^{(t+1)} = Y_{ik}^{(t)}$ whether $p_i\in\mathcal{W}_k$ or not.\\
\ignore{\begin{itemize}
\item $Y_{ik}^{(t+1)} = Y_{ik}^{(t)}$ whether $p_i\in\mathcal{W}_k$ or not.
\end{itemize}}
\textbf{Case 2}: $X_i = 1$ and $X_k=0$:
\begin{enumerate}
\item For $p_i\notin\mathcal{W}_k$, $Y_{ik}^{(t+1)} = Y_{ik}^{(t)} - \theta_{ki} \leq Y_{ik}^{(t)} - 1$.
\item For $p_i\in\mathcal{W}_k$, $Y_{ik}^{(t+1)} = Y_{ik}^{(t)} + \hat{\theta}_{ik}$.
\end{enumerate}
\textbf{Case 3}: $X_i = 0$ and $X_k=1$:
\begin{enumerate}
\item For $p_i\notin\mathcal{W}_k$, $Y_{ik}^{(t+1)} = Y_{ik}^{(t)} - \theta_{ik} \leq Y_{ik}^{(t)} - 1$.
\item For $p_i\in\mathcal{W}_k$, $Y_{ik}^{(t+1)} = Y_{ik}^{(t)} + \hat{\theta}_{ki}$.
\end{enumerate}
\textbf{Case 4}: $X_i = 1$ and $X_k=1$:
\begin{enumerate}
\item For $p_i\notin\mathcal{W}_k$, $Y_{ik}^{(t+1)} = Y_{ik}^{(t)} - \theta_{ki} - \theta_{ik} + 1 \leq Y_{ik}{(t)} - 1$.
\item For $p_i\in\mathcal{W}_k$, $Y_{ik}^{(t+1)} = Y_{ik}^{(t)} - 1$.
\end{enumerate}
The two inequalities in Condition 1 of both Cases 2 and 3 (i.e. when $p_i\notin\mathcal{W}_k$) arise from the following two facts:
 \begin{enumerate}
 \item $p_i$ is in $\mathcal{W}_i$ but not $\mathcal{W}_k$.
 \item Both receivers are targeted.
  \end{enumerate}
Fact 1 implies that there must exists at least one packet satisfying this condition and thus:
\begin{equation}
 \theta_{ik} = \psi_i - \psi_{ik} =  \left|\mathcal{W}_i\setminus \mathcal{W}_k\right| \geq 1.
\end{equation}
Fact 2 implies that $k$ is targeted by packet $p_k$ such that:
\begin{itemize}
\item $p_k\neq p_i$: or else $p_i$ must be in $\mathcal{W}_k$, which contradicts with Condition 1 of both Cases 2 and 3.
\item $p_k\notin \mathcal{W}_i$: or else the combination $p_i\oplus p_k$ will not be instantly decodable at $i$, which contradicts with the fact that $i$ is targeted.
\end{itemize}
Consequently, there exists at least one packet $p_k$ which is in $\mathcal{W}_k$ but not $\mathcal{W}_i$ and thus
\begin{equation}
\theta_{ki} = \psi_k - \psi_{ik} = \left|\mathcal{W}_k\setminus \mathcal{W}_i\right| \geq 1.
\end{equation}

The last inequality in Condition 1 of Case 4 for $p_i\notin\mathcal{W}_k$ arises from the fact that packet $p_i$ is in $\mathcal{W}_i$ but not in $\mathcal{W}_k$ and thus cannot be part of their intersection. Consequently, $\psi_i \geq \psi_{ik} + 1~\Rightarrow~\theta_{ik} \geq 1$ and $\psi_k \geq \psi_{ik} + 1~\Rightarrow~\theta_{ki} \geq 1$ and thus the left hand term will always be less than or equal to $Y_{ik}^{(t)} - 1$.

\section{Proof of \thref{th:expected-edge-set-size}} \label{app:graph-density}
To prove this theorem, we need to introduce this lemma, which is proved in \appref{app:expected-degree}.
\begin{lemma} \label{lem:expected-degree}
For any given feedback state\ignore{ $\boldsymbol{\varrho}$, $\boldsymbol{\psi}$ and $\boldsymbol{\psi}$ vectors}, the expected degree of any of the vertices induced by receiver $i$ (denoted by $\Delta_{i}$) is equal to:
\begin{equation}
\mathds{E}\left[\Delta_{i}\right] = \sum_{\substack{k = 1\\k\neq i}}^M\: \frac{\psi_k}{N} \left(1+\frac{\varrho_k\varrho_i}{N-1}\right)\;. \label{eq:primary_expected_degree}
\end{equation}
\end{lemma}
We will start our proof using \eqref{eq:sigma-delta}. When we ignore the contents of the different sets, we can derive an expression for the expected edge set size of the graph as follows:
\begin{align} \label{eq:expected-edge-set-size-proof-I}
\mathds{E}\left[\left|\mathcal{E}\right|\right] &= \frac{1}{2}\sum_{i=1}^M \mathds{E}\left[\Sigma \Delta_i\right] \ignore{\nonumber\\
&} = \frac{1}{2}\sum_{i=1}^M \sum_{\substack{k=1\\k\neq i}}^M \mathds{E}\left[\left|\mathcal{W}_i \cap \mathcal{W}_k\right|\right]  + \psi_k \mathds{E}\left[\sum_{j\in\mathcal{W}_i}I_{j\in\mathcal{H}_k}\right] \ignore{\nonumber \\
&} - \mathds{E}\left[\sum_{j\in\mathcal{W}_i}I_{j\in\mathcal{H}_k}\cdot\left|\mathcal{W}_k \cap \mathcal{W}_i\right|\right] \nonumber\\
& = \frac{1}{2} \sum_{i=1}^M \sum_{\substack{k=1\\k\neq i}}^M \frac{\psi_i\psi_k}{N}  + \frac{\psi_k \psi_i\varrho_k}{N} - \sum_{j\in\mathcal{W}_i} \mathds{E}\left[I_{j\in\mathcal{H}_k} \cdot \left|\mathcal{W}_k \cap \mathcal{W}_i\right|\right]
\end{align}
Note that the indicator function in the last term can be only zero or one. Consequently, the expectation of its multiplication with $\left|\mathcal{W}_k \cap \mathcal{W}_i\right|$ can be only evaluated for $I_{j\in\mathcal{H}_k} = 1$. In this case, packet $j$ cannot be in the intersection of $\mathcal{W}_k$ and $\mathcal{W}_i$. Consequently, this intersection is possible only with the other $\psi_i-1$ packets of receiver $i$, and from the set of the remaining $N-1$ packets. Since the cardinality of the intersection of two sets of given sizes, whose elements are drawn of the same pool of $N-1$ elements, is a hypergeometric random variable, we have:
\begin{align} \label{eq:expected-edge-set-size-proof-II}
& \mathds{E}\left[I_{j\in\mathcal{H}_k} \cdot \left|\mathcal{W}_k \cap \mathcal{W}_i\right|\right] = \sum_{n = 1}^{N-1}
n \: \mathds{P}\left[I_{j\in\mathcal{H}_k} = 1, \left|\mathcal{W}_k \cap \mathcal{W}_i\right| = n \right] \nonumber\\
&  = \sum_{n = 1}^{N-1}
n \: \mathds{P}\left[\left|\mathcal{W}_k \cap \mathcal{W}_i\right| = n \Big| I_{j\in\mathcal{H}_k} = 1\right] \cdot \mathds{P}\left[I_{j\in\mathcal{H}_k} = 1\right] \ignore{ \nonumber \\
&} = \sum_{n=1}^{N-1} n \frac{\binom{\psi_i-1}{n}\binom{N-1-\psi_i+1}{\psi_k - n}}{\binom{N-1}{\psi_k}} \:\frac{\varrho_k}{N} \ignore{\nonumber \\
&} = \frac{\varrho_k \psi_k \left(\psi_i - 1\right)}{N(N-1)}\;.
\end{align}
Substituting \eqref{eq:expected-edge-set-size-proof-II} in \eqref{eq:expected-edge-set-size-proof-I}, we get:
\begin{align}
\mathds{E}\left[\left|\mathcal{E}\right|\right] &= \frac{1}{2} \sum_{i=1}^M \sum_{\substack{k=1\\k\neq i}}^M \frac{\psi_i\psi_k}{N}  + \frac{\psi_k \psi_i\varrho_k}{N} - \frac{\psi_i \varrho_k \psi_k(\psi_i-1)}{N(N-1)} \ignore{\nonumber\\
&} = \frac{1}{2} \sum_{i=1}^M \psi_i \left\{ \sum_{\substack{k=1\\k\neq i}}^M \frac{\psi_k}{N} \left[ 1 + \varrho_k \left(1 - \frac{\psi_i-1}{N-1}\right)\right]\right\} \nonumber\\
& = \frac{1}{2} \sum_{i=1}^M \psi_i \left\{\sum_{\substack{k=1\\k\neq i}}^M \frac{\psi_k}{N} \left(1 + \frac{\varrho_k \varrho_i}{N-1}\right)\right\} = \frac{1}{2} \sum_{i=1}^M \psi_i \mathds{E}\left[\Delta_i\right]\;.
\end{align}

\section{Proof of \thref{th:expected-edge-set-size-evolution}} \label{app:expected-edge-set-size-evolution}
To prove this theorem, we first need to introduce the following lemma, which is proved in  \appref{app:degree-evolution}.
\begin{lemma} \label{lem:degree-evolution}
For a given maximal clique $\kappa$, chosen for transmission at time $t$, the expected degree of a receiver $i$ vertex at time $t+1$ is expressed, for $i\in\mathcal{T}$ and $i\notin\mathcal{T}$, as:
\ignore{
\begin{align}
\mathds{E}\left[\Delta_{i\in\mathcal{T}}^{(t+1)}\right]  = & \mathds{E}\left[\Delta_i^{(t)}\right] + \sum_{\substack{k=1\\k\neq i}}^M q_i \xi_k\ignore{\nonumber \\
&} -\sum_{\substack{k\in\mathcal{T}\\k\neq i}} \Phi_k(q_i) \label{eq:degree-evolution-1}
\end{align}
if $i\in\mathcal{T}$ and is expressed as:
\begin{equation}
\mathds{E}\left[\Delta_{i\notin\mathcal{T}}^{(t+1)}\right]  = \mathds{E}\left[\Delta_i^{(t)}\right] -\sum_{k\in\mathcal{T}} \Phi_k(0) \label{eq:degree-evolution-2}
\end{equation}
if $i\notin\mathcal{T}$.
}
\begin{align}
\mathds{E}\left[\Delta_{i\in\mathcal{T}}^{(t+1)}\right]  = & \mathds{E}\left[\Delta_i^{(t)}\right] + \sum_{\substack{k=1\\k\neq i}}^M q_i \xi_k -\sum_{\substack{k\in\mathcal{T}\\k\neq i}} \Phi_k(q_i)\ \label{eq:degree-evolution-1}\\
\mathds{E}\left[\Delta_{i\notin\mathcal{T}}^{(t+1)}\right]  = & \mathds{E}\left[\Delta_i^{(t)}\right] -\sum_{k\in\mathcal{T}} \Phi_k(0) \label{eq:degree-evolution-2}
\end{align}
\end{lemma}
When the maximal clique $\kappa$ is chosen for transmission at time $t$, the receivers in $\mathcal{T}$ is targeted with the coded packet but are not guaranteed to receive that packet. From \lref{th:expected-edge-set-size}, we can derive the expression of the expected edge set size at time $t+1$, conditioned on the random vector $\mathbf{X}$ defined in \appref{app:degree-evolution}, as follows:
\begin{equation}
\mathds{E}\left[\left|\mathcal{E}^{(t+1)}\right|\Big|\mathbf{X}\right] =  \frac{1}{2} \sum_{i\in\mathcal{T}} (\psi_i - X_i)\; \mathds{E}\left[\Delta_{i\in\mathcal{T}}^{(t+1)}\Big|\mathbf{X}\right] \ignore{\nonumber\\
&} + \frac{1}{2} \sum_{i\notin\mathcal{T}} \psi_i \; \mathds{E}\left[\Delta_{i\notin\mathcal{T}}^{(t+1)}\Big|\mathbf{X}\right] \;.
\end{equation}
Now, taking the expectation operator over the random vector $\mathbf{X}$, we can get the expression for the expected edge set size at time $t+1$ as follows:
\begin{align} \label{eq:expected-edge-set-size-evolution-proof-0}
&\mathds{E}\left[\left|\mathcal{E}^{(t+1)}\right|\right] = \mathds{E}_{\mathbf{X}} \left\{ \mathds{E}\left[\left|\mathcal{E}^{(t+1)}\right|\Big|\mathbf{X}\right]\right\} \nonumber \\
& =  \frac{1}{2} \sum_{i\in\mathcal{T}} \psi_i\; \mathds{E}_{\mathbf{X}}\left\{\mathds{E}\left[\Delta_{i\in\mathcal{T}}^{(t+1)}\Big|\mathbf{X}\right]\right\} \ignore{\nonumber\\
&} - \frac{1}{2}\sum_{i\in\mathcal{T}} \mathds{E}_{\mathbf{X}}\left\{X_i \; \mathds{E}\left[\Delta_{i\in\mathcal{T}}^{(t+1)}\Big|\mathbf{X}\right]\right\} \ignore{\nonumber\\
&}+ \frac{1}{2} \sum_{i\notin\mathcal{T}} \psi_i \; \mathds{E}_{\mathbf{X}}\left\{\mathds{E}\left[\Delta_{i\notin\mathcal{T}}^{(t+1)}\Big|\mathbf{X}\right]\right\}\ignore{\;.}
\end{align}
From \eqref{eq:degree-evolution-proof-I}, we know that:
\begin{equation} \label{eq:expected-edge-set-size-evolution-proof-I}
\mathds{E}_{\mathbf{X}}\left\{\mathds{E}\left[\Delta_{i\in\mathcal{T}}^{(t+1)}\Big|\mathbf{X}\right]\right\} = \mathds{E}\left[\Delta_{i\in\mathcal{T}}^{(t+1)}\right] \ignore{\nonumber\\
&} = \mathds{E}\left[\Delta_i^{(t)}\right] + \sum_{\substack{k=1\\k\neq i}}^M q_i \xi_k\ignore{\nonumber \\
&} -\sum_{\substack{k\in\mathcal{T}\\k\neq i}} \Phi_k(q_i) \ignore{\nonumber\\
&} = \mathds{E}\left[\Delta_i^{(t)}\right] + \alpha_i\;.
\end{equation}
Similarly, we know that\ignore{ from \eqref{eq:degree-evolution-proof-II} that}:
\begin{equation} \label{eq:expected-edge-set-size-evolution-proof-II}
\mathds{E}_{\mathbf{X}}\left\{\mathds{E}\left[\Delta_{i\notin\mathcal{T}}^{(t+1)}\Big|\mathbf{X}\right]\right\} = \mathds{E}\left[\Delta_{i\notin\mathcal{T}}^{(t+1)}\right]\ignore{\nonumber\\
& }= \mathds{E}\left[\Delta_i^{(t)}\right] -\sum_{\substack{k\in\mathcal{T}\\k\neq i}} \Phi_k(0) \ignore{\nonumber\\
&} = \mathds{E}\left[\Delta_i^{(t)}\right] + \beta_i\;.
\end{equation}
From \eqref{eq:degree-evolution-proof-I}, we can compute $\mathds{E}_{\mathbf{X}}\left\{X_i \; \mathds{E}\left[\Delta_{i\in\mathcal{T}}^{(t+1)}\Big|\mathbf{X}\right]\right\}$ as follows:
\begin{align}\label{eq:expected-edge-set-size-evolution-proof-III}
\mathds{E}_{\mathbf{X}}\left\{X_i \; \mathds{E}\left[\Delta_{i\in\mathcal{T}}^{(t+1)}\Big|\mathbf{X}\right]\right\}
= \;& \mathds{E}_{\mathbf{X}}\left\{X_i\right\}\mathds{E}\left[\Delta_i^{(t)}\right] \ignore{\nonumber \\
&} + \sum_{\substack{k=1\\k\neq i}}^M  \frac{\psi_k\varrho_k \mathds{E}_{\mathbf{X}}\left\{X_i^2\right\}}{N(N-1)} \nonumber\\
& - \sum_{\substack{k\in\mathcal{T}\\k\neq i}}  \mathds{E}_{\mathbf{X}}\left\{\frac{X_i X_k}{N} \ignore{\nonumber\\
& \qquad\qquad\qquad~} + \frac{\left(X_k\left(\varrho_k - \psi_k\right) + X_k^2\right)\left(\varrho_i X_i + X^2_i\right)}{N(N-1)}\right\}
\end{align}
Using the definitions of $\xi_k$, $\Phi_k(x)$ and $\gamma_i$ in \eqref{eq:xi}\ignore{, \eqref{eq:phi}} and \eqref{eq:gamma}, respectively, we get:
\begin{equation}
\mathds{E}_{\mathbf{X}}\left\{X_i \; \mathds{E}\left[\Delta_{i\in\mathcal{T}}^{(t+1)}\Big|\mathbf{X}\right]\right\}  = q_i \mathds{E}\left[\Delta_i^{(t)}\right] + \sum_{\substack{k=1\\k\neq i}}^M q_i\xi_k  \ignore{\nonumber \\
&}-\sum_{\substack{k\in\mathcal{T}\\k\neq i}}q_i\Phi_k(1) \ignore{\nonumber\\
 = \;&} = q_i \left(\mathds{E}\left[\Delta_i^{(t)}\right] + \gamma_i\right)\;.
\end{equation}
Finally, we know that:
\begin{equation} \label{eq:expected-edge-set-size-evolution-proof-IV}
\frac{1}{2}\sum_{i\in\mathcal{T}} \psi_i \mathds{E}\left[\Delta_i^{(t)}\right] + \frac{1}{2}\sum_{i\notin\mathcal{T}} \psi_i \mathds{E}\left[\Delta_i^{(t)}\right] = \mathds{E}\left[\left|\mathcal{E}^{(t)}\right|\right]
\end{equation}
The theorem follows by substituting \eqref{eq:expected-edge-set-size-evolution-proof-I}, \eqref{eq:expected-edge-set-size-evolution-proof-II} and \eqref{eq:expected-edge-set-size-evolution-proof-III} in \eqref{eq:expected-edge-set-size-evolution-proof-0}, re-arranging the terms and finally substituting \eqref{eq:expected-edge-set-size-evolution-proof-IV} in the re-arranged equation.

\section{Proof of \thref{th:degree-comparison}} \label{app:degree-comparison}
Note that $\psi_i > \psi_{h}$ implicitly means that $\varrho_h < \varrho_i$. From \lref{lem:expected-degree} in \appref{app:graph-density}, we have:
\begin{equation} \label{eq:degree-comparison}
\mathds{E}\left[\Delta_h\right] = \sum_{\substack{k = 1\\k\neq i,h}}^M\: \frac{\psi_k}{N} \left(1+\frac{\varrho_k\varrho_h}{N-1}\right) + \frac{\psi_i}{N} \left(1+\frac{\varrho_i\varrho_h}{N-1}\right)\ignore{ \nonumber \\
> &} > \sum_{\substack{k = 1\\k\neq i,h}}^M\: \frac{\psi_k}{N} \left(1+\frac{\varrho_k\varrho_i}{N-1}\right) \ignore{\nonumber \\
&}+ \frac{\psi_h}{N} \left(1+\frac{\varrho_i\varrho_h}{N-1}\right) \ignore{\nonumber \\
 =&} = \mathds{E}\left[\Delta_i\right]\;.
\end{equation}

\section{Proof of \thref{th:alpha-beta-gamma}} \label{app:alpha-beta-gamma}
From \ignore{Equations \eqref{eq:alpha}, \eqref{eq:beta},}\eqref{eq:gamma}\ignore{, \eqref{eq:phi},} and \eqref{eq:xi}, we have:
\begin{align}
\alpha_i - \frac{q_i\gamma_i}{\psi_i} = &  \sum_{\substack{k=1\\k\neq i}}^M  \frac{q_i\psi_k \varrho_k}{N(N-1)} \ignore{\nonumber\\
&} -\sum_{\substack{k\in\mathcal{T}\\k\neq i}}\frac{q_k}{N} \left( 1 + \frac{\left(\varrho_k-\psi_k+1\right) \left(\varrho_i+q_i\right)}{N-1}\right) \ignore{\nonumber \\
&} - \sum_{\substack{k=1\\k\neq i}}^M \frac{q_i\psi_k \varrho_k}{N\psi_i(N-1)} \nonumber\\
& + \sum_{\substack{k\in\mathcal{T}\\k\neq i}} \frac{q_iq_k}{N\psi_i} \left( 1 + \frac{\left(\varrho_k-\psi_k+1\right) \left(\varrho_i+1\right)}{N-1}\right) \nonumber \\ \ignore{\;.
\end{align}
Re-arranging the above expression and using the definition of $\beta_i$ in \ignore{\eqref{eq:beta}}\eqref{eq:gamma}, we get:
\begin{align}
\alpha_i - \frac{q_i\gamma_i}{\psi_i} }=& \; \beta_i \ignore{\nonumber \\
&} + \sum_{\substack{k\notin\mathcal{T}\\k\neq i}}  \frac{q_i\psi_k \varrho_k}{N(N-1)}\left(1-\frac{1}{\psi_i}\right) \ignore{\nonumber\\
&} +\sum_{\substack{k\in\mathcal{T}\\k\neq i}}\frac{q_i\psi_k\varrho_k \left(1-\frac{1}{\psi_i}\right)- q_iq_k\left(\varrho_k-\psi_k+1\right)}{N(N-1)} \nonumber \\
& + \sum_{\substack{k\in\mathcal{T}\\k\neq i}} \frac{q_iq_k}{N\psi_i} \left( 1 + \frac{\left(\varrho_k-\psi_k+1\right) \left(\varrho_i+1\right)}{N-1}\right) \;.\label{eq:main_I}
\end{align}
Since $1-\frac{1}{\psi_i} \geq 0$, as long as $\psi_i>0$, then:
\begin{equation} \label{eq:cond_I}
\sum_{\substack{k\notin\mathcal{T}\\k\neq i}}  \frac{q_i\psi_k \varrho_k}{N(N-1)}\left(1-\frac{1}{\psi_i}\right) \geq 0\;.
\end{equation}
Since $\varrho_k \geq \psi_k~\forall~k\in\mathcal{T},k\neq i$, we get:
\begin{equation}\label{eq:cond_II}
\sum_{\substack{k\in\mathcal{T}\\k\neq i}} \frac{q_iq_k}{N\psi_i} \left( 1 + \frac{\left(\varrho_k-\psi_k+1\right) \left(\varrho_i+1\right)}{N-1}\right) \geq 0\;.
\end{equation}
For $\psi_i>1$ and $\psi_k>1,~\forall~k\in\mathcal{T}\setminus i$, we have $\varrho_k \geq \varrho_k - \psi_k +1$ and $\psi_k\left(1-\frac{1}{\psi_i}\right) \geq q_k$. Consequently\ignore{, we have}:
\begin{equation} \label{eq:cond_III}
\sum_{\substack{k\in\mathcal{T}\\k\neq i}}\frac{q_i\psi_k\varrho_k \left(1-\frac{1}{\psi_i}\right)- q_iq_k\left(\varrho_k-\psi_k+1\right)}{N(N-1)} \geq 0
\end{equation}
Note that this expression becomes much greater than zero, when all $\psi_k\in\mathcal{T},k\neq i$ are much greater than one (i.e. when all targeted vertices other than $i$ have the largest Wants sets). Finally, in case $\psi_i = 1$ and $\psi_k>1,~\forall~k\in\mathcal{T}\setminus i$, we have the third summation in \eqref{eq:main_I} greater than the negative term in the second summation in \eqref{eq:main_I}. The theorem follows from all these inequalities\ignore{in \eqref{eq:cond_I}, \eqref{eq:cond_II} and \eqref{eq:cond_III}} and \eqref{eq:main_I}.

\section{Proof of \lref{lem:expected-degree}} \label{app:expected-degree}
We will start this proof from the exact vertex degree expression in \eqref{eq:vertex-degree}. Ignoring the content of the different sets, we can derive the expression for the expected degree of a vertex of receiver $i$ as follows:
\begin{align} \label{eq:lem-I-proof-I}
\mathds{E}\left[\Delta_i\right] & = \mathds{E}\left[\Delta_{ij}\right] \ignore{\nonumber\\
&}= \sum_{\substack{k=1\\k\neq i}}^M \mathds{E}\left[I_{j\in\mathcal{W}_k}\right] + \mathds{E}\left[I_{j\in\mathcal{H}_k}\right]\left|\mathcal{W}_k\right| - \mathds{E}\left[I_{j\in\mathcal{H}_k} \cdot \left|\mathcal{W}_k \cap \mathcal{W}_i\right|\right] \nonumber\\
&= \sum_{\substack{k=1\\k\neq i}}^M \frac{\psi_k}{N} + \frac{\varrho_k \psi_k}{N} - \mathds{E}\left[I_{j\in\mathcal{H}_k} \cdot \left|\mathcal{W}_k \cap \mathcal{W}_i\right|\right] \;.
\end{align}

Substituting \eqref{eq:expected-edge-set-size-proof-II} in \eqref{eq:lem-I-proof-I}, we get:
\begin{equation}
\mathds{E}\left[\Delta_i\right] = \sum_{\substack{k=1\\k\neq i}}^M \frac{\psi_k}{N} \left[ 1 + \varrho_k \left(1 - \frac{\psi_i-1}{N-1}\right)\right] \ignore{\nonumber \\
&} = \sum_{\substack{k=1\\k\neq i}}^M \frac{\psi_k}{N} \left( 1 + \frac{\varrho_k \varrho_i}{N-1}\right)\;.
\end{equation}

\section{Proof of \lref{lem:degree-evolution}} \label{app:degree-evolution}
When the maximal clique $\kappa$ is chosen for transmission at time $t$, each member $k$ of the targeted receiver set $\mathcal{T}$ may receive the coded packet with probability $q_k$. Let $X_k$ be the random variable representing the reception of receiver $k\in\mathcal{T}$ at time $t$ and $\mathbf{X}$ as the random vector of all such random variables. From \lref{lem:expected-degree}, we can derive the expression of the expected degree of receiver $i\in\mathcal{T}$ at time $t+1$, conditioned on the random vector $\mathbf{X}$, as follows:
\begin{align}
\mathds{E}\left[\Delta_{i\in\mathcal{T}}^{(t+1)}\Big|\mathbf{X}\right]
&= \sum_{\substack{k\in\mathcal{T}\\k\neq i}} \frac{\psi_k - X_k}{N}\left(1 + \frac{\left(\varrho_k+X_k\right)\left(\varrho_i + X_i\right)}{N-1}\right) \ignore{\nonumber\\
&} + \sum_{k\notin\mathcal{T}} \frac{\psi_k}{N}\left(1 + \frac{\varrho_k\left(\varrho_i + X_i\right)}{N-1}\right) \label{eq:reference-I}\\
& = \ignore{\sum_{\substack{k=1\\k\neq i}}^M \frac{\psi_k}{N}\left(1 + \frac{\varrho_k \varrho_i}{N-1}\right)}
 \mathds{E}\left[\Delta_i^{(t)}\right]+ \sum_{\substack{k=1\\k\neq i}}^M  \frac{\psi_k\varrho_k X_i}{N(N-1)} \ignore{ \nonumber\\
&} - \sum_{\substack{k\in\mathcal{T}\\k\neq i}} \frac{X_k}{N}\left(1 + \frac{\left(\varrho_k -\psi_k + X_k\right)\left(\varrho_i + X_i\right)}{N-1}\right)\label{eq:reference-II}\;.
\end{align}
\ignore{The first term in \eqref{eq:reference-II} is obviously the expected vertex degree of receiver $i$ at time $t$.} Now, we can derive the expected degree of receiver $i$ after serving the maximal clique $\kappa$ as follows:
\begin{align}\label{eq:degree-evolution-proof-I}
\mathds{E}\left[\Delta_{i\in\mathcal{T}}^{(t+1)}\right]  &= \mathds{E}_{\mathbf{X}}\left\{\mathds{E}\left[\Delta_{i\in\mathcal{T}}^{(t+1)}\Big|\mathbf{X}\right]\right\} \nonumber\\
& = \mathds{E}\left[\Delta_i^{(t)}\right] + \sum_{\substack{k=1\\k\neq i}}^M  \frac{\psi_k\varrho_k \mathds{E}_{\mathbf{X}}\left\{X_i\right\}}{N(N-1)} \ignore{\nonumber\\
& } - \sum_{\substack{k\in\mathcal{T}\\k\neq i}} \mathds{E}_{\mathbf{X}}\left\{\frac{X_k}{N} + \frac{\left(X_k\left(\varrho_k - \psi_k\right) + X_k^2\right)\left(\varrho_i + X_i\right)}{N(N-1)}\right\} \;.
\end{align}
Using the definition of $\xi_k$ and $\Phi_k(x)$ in \eqref{eq:xi}\ignore{and \eqref{eq:phi}}, respectively, we get:
\begin{equation}
\mathds{E}\left[\Delta_{i\in\mathcal{T}}^{(t+1)}\right]  = \mathds{E}\left[\Delta_i^{(t)}\right] + \sum_{\substack{k=1\\k\neq i}}^M q_i \xi_k -\sum_{\substack{k\in\mathcal{T}\\k\neq i}} \Phi_k(q_i)\;.
\end{equation}
The expression for $\mathds{E}\left[\Delta_{i\notin\mathcal{T}}^{(t+1)}\right]$ can be similarly derived using a similar approach.

\bibliographystyle{IEEEtran}
\bibliography{IEEEabrv,bibfile}

\ignore{
\newpage

\begin{table}
\centering
\caption{Evolution cases of the Wants sets of a pair of receivers when the sender transmits a packet combination $\kappa$ targeting a set of receivers $\mathcal{T}$}\label{tab:evolution-conditions}
\begin{tabular}{|c|c|c|c|c|c|}
  \hline
  \multicolumn{3}{|c|}{Cases} & \multicolumn{1}{|c|}{\multirow{2}{*}{$\psi_i$}} & \multicolumn{1}{|c|}{\multirow{2}{*}{$\psi_k$}} & \multicolumn{1}{|c|}{\multirow{2}{*}{$\psi_{ik}$}} \\ \cline{1-3}
  $i$ & $k$ & $p_i/p_k$ & \multicolumn{1}{|c|}{} & \multicolumn{1}{|c|}{} & \multicolumn{1}{|c|}{} \\ \hline \hline
  $\notin\mathcal{T}$ & $\notin\mathcal{T}$ & - & $\psi_i$ & $\psi_k$ & $\psi_{ik}$ \\ \hline
  $\notin\mathcal{T}$ & $\in\mathcal{T}$ & $p_k\notin\mathcal{W}_i$ & $\psi_i$ & $\psi_k - X_k$ & $\psi_{ik}$ \\ \hline
  $\notin\mathcal{T}$ & $\in\mathcal{T}$ & $p_k\in\mathcal{W}_i$ & $\psi_i$ & $\psi_k - X_k$ & $\psi_{ik} - X_k$ \\ \hline
  $\in\mathcal{T}$ & $\notin\mathcal{T}$ & $p_i\notin\mathcal{W}_k$ & $\psi_i - X_i$ & $\psi_k$ & $\psi_{ik}$ \\ \hline
  $\in\mathcal{T}$ & $\notin\mathcal{T}$ & $p_i\in\mathcal{W}_k$ & $\psi_i - X_i$ & $\psi_k$ & $\psi_{ik} - X_i$ \\ \hline
  $\in\mathcal{T}$ & $\in\mathcal{T}$ & $p_i\notin\mathcal{W}_k$ & $\psi_i - X_i$ & $\psi_k - X_k$ & $\psi_{ik}$ \\ \hline
  $\in\mathcal{T}$ & $\in\mathcal{T}$ & $p_i\in\mathcal{W}_k$ & $\psi_i - X_i$ & $\psi_k - X_k$ & $\psi_{ik} - X_{ik}$ \\ \hline
\end{tabular}
\end{table}
}

\end{document}